\documentclass[a4paper,onecolumn,11pt,accepted=2021-10-12]{quantumarticle}
\pdfoutput=1

\usepackage[font=small]{caption}
\usepackage{subcaption}

\usepackage{amsmath,amsfonts,amssymb,amsthm}
\usepackage[margin=1in]{geometry}
\usepackage[colorlinks]{hyperref}
\usepackage{mathrsfs}
\usepackage{mathtools}
\usepackage{thmtools}
\usepackage{thm-restate}
\mathtoolsset{centercolon}
\usepackage[numbers,comma,sort&compress]{natbib}
\usepackage[linesnumbered,ruled]{algorithm2e}
\usepackage{bm}
\usepackage{bookmark}
\usepackage{multirow}
\usepackage{xcolor}

\usepackage{enumerate}

\let\originalleft\left
\let\originalright\right
\renewcommand{\left}{\mathopen{}\mathclose\bgroup\originalleft}
\renewcommand{\right}{\aftergroup\egroup\originalright}

\newcommand{\C}{{\mathbb{C}}}

\newcommand{\N}{{\mathbb{N}}}
\newcommand{\R}{{\mathbb{R}}}
\newcommand{\Z}{{\mathbb{Z}}}

\newcommand{\D}{{\mathscr{D}}}
\renewcommand{\L}{\mathscr{L}}

\renewcommand{\d}{\mathrm{d}}

\newcommand{\Poisson}{\mathrm{Poisson}}

\newcommand{\ket}[1]{|{#1}\rangle}

\newcommand{\norm}[1]{\|{#1}\|}

\newcommand{\range}[1]{[{#1}]}
\newcommand{\rangez}[1]{[{#1}]_0}

\DeclareMathOperator{\poly}{poly}

\newtheorem{lemma}{Lemma}
\newtheorem{corollary}{Corollary}
\newtheorem{problem}{Problem}

\numberwithin{equation}{section}

\newcommand{\eq}[1]{(\ref{eq:#1})}
\renewcommand{\sec}[1]{\hyperref[sec:#1]{Section~\ref*{sec:#1}}}
\newcommand{\app}[1]{\hyperref[app:#1]{Appendix~\ref*{app:#1}}}
\newcommand{\thm}[1]{\hyperref[thm:#1]{Theorem~\ref*{thm:#1}}}
\newcommand{\lem}[1]{\hyperref[lem:#1]{Lemma~\ref*{lem:#1}}}
\newcommand{\cor}[1]{\hyperref[cor:#1]{Corollary~\ref*{cor:#1}}}
\newcommand{\prb}[1]{\hyperref[prb:#1]{Problem~\ref*{prb:#1}}}
\newcommand{\tab}[1]{\hyperref[tab:#1]{Table~\ref*{tab:#1}}}

\newcommand{\gdd}{global strict diagonal dominance}

\makeatletter
\let\@@magyar@captionfix\relax
\newcommand\notsotiny{\@setfontsize\notsotiny{7.5}{8.5}}
\makeatother

\begin{document}


\title{High-precision quantum algorithms for \newline partial differential equations}

\author[1,2,3]{Andrew M.\ Childs}
\author[1,2,4]{Jin-Peng Liu}
\author[1,2,5]{Aaron Ostrander}
\affil[1]{Joint Center for Quantum Information and Computer Science, University of Maryland, MD 20742, USA}
\affil[2]{Institute for Advanced Computer Studies, University of Maryland, MD 20742, USA}
\affil[3]{Department of Computer Science, University of Maryland, MD 20742, USA}
\affil[4]{Department of Mathematics, University of Maryland, MD 20742, USA}
\affil[5]{Department of Physics, University of Maryland, MD 20742, USA}

\date{}
\maketitle


\begin{abstract}
Quantum computers can produce a quantum encoding of the solution of a system of differential equations exponentially faster than a classical algorithm can produce an explicit description. However, while high-precision quantum algorithms for linear ordinary differential equations are well established, the best previous quantum algorithms for linear partial differential equations (PDEs) have complexity $\poly(1/\epsilon)$, where $\epsilon$ is the error tolerance. By developing quantum algorithms based on adaptive-order finite difference methods and spectral methods, we improve the complexity of quantum algorithms for linear PDEs to be $\poly(d, \log(1/\epsilon))$, where $d$ is the spatial dimension. Our algorithms apply high-precision quantum linear system algorithms to systems whose condition numbers and approximation errors we bound. We develop a finite difference algorithm for the Poisson equation and a spectral algorithm for more general second-order elliptic equations.
\end{abstract}


\section{Introduction}
\label{sec:problem}

Many scientific problems involve partial differential equations (PDEs). Prominent examples include Maxwell's equations for electromagnetism, Boltzmann's equation and the Fokker-Planck equation in thermodynamics, and Schr{\"o}dinger's equation in continuum quantum mechanics. While models of physics are often studied in a constant number of spatial dimensions, it is also natural to study high-dimensional PDEs, such as to model systems with many interacting particles. Classical numerical methods have complexity that grows exponentially in the dimension, a phenomenon sometimes called the curse of dimensionality~\cite{Bel57}. This is a major challenge for attempts to solve PDEs on classical computers.

A common approach to solving PDEs on a digital computer is the finite difference method (FDM). In this approach, we discretize space into a rectangular lattice, solve a system of linear equations that approximates the PDE on the lattice, and output the solution on those grid points. If each spatial coordinate has $n$ discrete values, then $n^d$ points are needed to discretize a $d$-dimensional problem. Simply outputting the solution on these grid points takes time $\Omega(n^d)$.

Beyond uniform grids, the sparse grid technique \cite{Smo63} has been applied to reduce the time and space complexity of outputting a sparse encoding of the solution to $O(n\log^d n)$ \cite{BG04,Zen91}. While this is a significant improvement, it still scales exponentially in $d$. It can be shown that for a grid-based approach this complexity is optimal with respect to certain norms \cite{BG04}. Reference~\cite{BG04} proposes alternative sparse grid algorithms whose complexities scale linearly with $n$ but exponentially with $d$. Another grid-based method is the finite element method (FEM), where the differential equation is multiplied by functions with local support (restricted by the grid) and then integrated. This produces a set of equations that the solution must satisfy, which are then used to approximate the solution. In yet another grid-based approach, the finite volume method (FVM) considers a grid dividing space into volumes/cells. The field is integrated over these volumes to create auxiliary variables, and relations between these variables are derived from the differential equation.

An alternative to grid methods is the concept of \emph{spectral methods} \cite{Ghe07,STW11}. Spectral methods use linear combinations of basis functions (such as Fourier basis states or Chebyshev polynomials) to globally approximate the solution. These basis functions allow the construction of a linear system whose solution approximates the solution of the PDE.

These classical algorithms often consider the problem of outputting the solution at $N$ points in space, which clearly requires $\Omega(N)$ space and time. Quantum algorithms often (though not always) consider the alternative problem of outputting a quantum state proportional to such a vector, which requires only $\Omega(\log N)$ space---and correspondingly provides more limited access to the solution---but can potentially be done in only $\poly(\log N)$ time.

The fact that quantum states can efficiently encode exponentially long vectors has also been leveraged for the development of quantum linear system algorithms (QLSAs) \cite{HHL08,Amb12,CKS15}. For a linear system $A \vec x = \vec b$, a QLSA outputs a quantum state proportional to the solution $\vec x$. To learn information about the solution $\vec x$, the output of the QLSA must be post-processed. For example, to output all the entries of an $N$-dimensional vector $\vec x$ given a quantum state $\ket{x}$ proportional to it, even a quantum computer needs time and space $\Omega(N)$.

Because linear systems are often used in classical algorithms for PDEs such as those described above, it is natural to consider their quantum counterparts. Clader, Jacobs, and Sprouse~\cite{CJS13} give a heuristic algorithm for using sparse preconditioners and QLSAs to solve a linear system constructed using the FEM for Maxwell's equations. The state output by the QLSA is then post-processed to compute electromagnetic scattering cross-sections.

In subsequent work, Montanaro and Pallister~\cite{MP16} use QLSAs to implement the FEM for $d$-dimensional boundary value problems and evaluate the quantum speedup that can be achieved when estimating a function of the solution within precision $\epsilon$. This involves a careful analysis of how different algorithmic parameters (such as the dimension and condition number of the FEM linear system and the number of post-processing measurements) scale with respect to input variables (such as the spatial dimension $d$ and desired precision $\epsilon$), since all of these affect the complexity. Their algorithms have complexity $\poly(d, 1/\epsilon)$, compared to $O((1/\epsilon)^d)$ for the classical FEM. This exponential improvement with respect to $d$ suggests that quantum algorithms may be notably faster when $d$ is large. However, they also argue that for fixed $d$, at most a polynomial speed-up can be expected due to lower bounds on the cost of post-processing the state to estimate a function of the solution.

The FDM has also been used in quantum algorithms for PDEs. References~\cite{cao2013quantum,wang2019quantum} apply the FDM to solve Poisson's equation in rectangular volumes under Dirichlet boundary conditions. Although the circuits they construct have $\poly (\log (1/ \epsilon))$ gates, these circuits have success probability $\poly(1/\epsilon)$, leading to $\poly(1/\epsilon)$ time complexity. Additionally, they do not quantify errors resulting from the finite-difference approximation. Reference~\cite{CJO17} applies the FDM to the problem of outputting states proportional to solutions of the wave equation, giving complexity $d^{\frac{5}{2}}\poly(1/\epsilon)$, a polynomial dependence on $d$ and $1/\epsilon$ (which is $\poly(n)$ for a fixed-order FDM). The FVM is combined with the reservoir method in Reference~\cite{Fillion-Gourdeau2018} to simulate hyperbolic equations; although they achieve linear scaling with respect to the spatial dimension, they use fixed order differences, leading to $\poly (1/ \epsilon )$ scaling. These FDM, FEM, and FVM approaches can only give a total complexity $\poly(1/\epsilon)$, even using high-precision methods for the QLSA or Hamiltonian simulation, because of the additional approximation errors in the FDM, FEM, and FVM.

The FDM is also applied in Reference~\cite{kivlichan2017bounding} to simulate how a fixed number of particles evolve under the Schr\"{o}dinger equation with access to an oracle for the potential term. This can be seen as a special case of quantum algorithms for PDEs. Other examples include quantum algorithms for many-body quantum dynamics \cite{zalka1998efficient,wiesner1996simulations} and for electronic structure problems, including for quantum chemistry (see for example References~\cite{Reiher7555, lanyon2010towards}). However, here we focus on PDEs whose dynamics are not necessarily unitary.

In this paper, we propose new quantum algorithms for linear PDEs where the boundary is the unit hypercube. In the spirit of Reference~\cite{MP16}, we state our results in terms of the approximation error and the spatial dimension; however, we do not consider the problem of estimating a function of the PDE solution and instead focus on outputting states encoding the solution, allowing us to give algorithms with complexity $\poly(\log(1/\epsilon))$. Just as for the QLSA, this improvement is potentially significant if the given equations must be solved as a subroutine within some larger computation. The problem we address can be informally stated as follows: Given a linear PDE with boundary conditions and an error parameter $\epsilon$, output a quantum state that is $\epsilon$-close to one whose amplitudes are proportional to the solution of the PDE at a set of grid points in the domain of the PDE.
We focus on elliptic PDEs, and we assume a technical condition that we call \emph{global strict diagonal dominance} (defined in \eq{DDM}).

Our first algorithm is based on a quantum version of the FDM approach: we use a finite-difference approximation to produce a system of linear equations and then solve that system using the QLSA. We analyze our FDM algorithm as applied to Poisson's equation (which automatically satisfies global strict diagonal dominance) under periodic, Dirichlet, and Neumann boundary conditions.
Whereas previous FDM approaches \cite{cao2013quantum,CJO17} considered fixed orders of truncation, we adapt the order of truncation depending on $\epsilon$, inspired by the classical adaptive FDM \cite{BG87}.
As the order increases, the eigenvalues of the FDM matrix approach the eigenvalues of the continuous Laplacian, allowing for more precise approximations.
The main algorithm we present uses the quantum Fourier transform (QFT) and takes advantage of the high-precision LCU-based QLSA \cite{CKS15}. We first consider periodic boundary conditions, but by restricting to appropriate subspaces, this approach can also be applied to homogeneous Dirichlet and Neumann boundary conditions.
We state our result in \thm{main_fdm}, which (informally) says that this quantum adaptive FDM approach produces a quantum state approximating the solution of Poisson's equation with complexity $d^{6.5}\poly(\log d, \log(1/\epsilon))$.

We also propose a quantum algorithm for more general second-order elliptic PDEs under periodic or non-periodic Dirichlet boundary conditions. This algorithm is based on quantum spectral methods \cite{CL19}. The spectral method globally approximates the solution of a PDE by a truncated Fourier or Chebyshev series (which converges exponentially for smooth functions) with undetermined coefficients, and then finds the coefficients by solving a linear system. This system is exponentially large in $d$, so solving it is infeasible for classical algorithms but feasible in a quantum context. To be able to apply the QLSA efficiently, we show how to make the system sparse using variants of the quantum Fourier transform.
Our bound on the condition number of the linear system uses global strict diagonal dominance, and introduces a factor in the complexity that measures the extent to which this condition holds.
We state our result in \thm{main}, which (informally) gives a complexity of $d^{2}\poly(\log(1/\epsilon))$ for producing a quantum state approximating the solution of general second-order elliptic PDEs with Dirichlet boundary conditions.

Both of these approaches have complexity $\poly(d, \log(1/\epsilon))$, providing optimal dependence on $\epsilon$ and an exponential improvement over classical methods as a function of the spatial dimension $d$. Bounding the complexities of these algorithms requires analyzing how $d$ and $\epsilon$ affect the condition numbers of the relevant linear systems (finite difference matrices and matrices relating the spectral coefficients) and accounting for errors in the approximate solution provided by the QLSA. Furthermore, the complexities of both approaches scale logarithmically with high-order derivatives of the solution and the inhomogeneity. The detailed complexity dependence is presented in \thm{main_fdm} and \thm{main}, and is further discussed in \sec{discussion}.

\tab{alg-compare} compares the performance of our approaches to other classical and quantum algorithms for PDEs. Compared to classical algorithms, quantum algorithms improve the dependence on spatial dimension from exponential to polynomial (with the significant caveat that they produce a different representation of the solution). Compared to previous quantum FDM/FEM/FVM algorithms \cite{cao2013quantum,Fillion-Gourdeau2018,CJO17,MP16}, the quantum adaptive FDM and quantum spectral method improve the error dependence from $\poly(1/\epsilon)$ to $\poly(\log(1/\epsilon))$. Our approaches achieve the best known dependence on the parameter $\epsilon$ for the Poisson equation with homogeneous boundary conditions. Furthermore, our quantum spectral method approach not only achieves the best known dependence on $d$ and $\epsilon$ for elliptic PDEs with inhomogeneous Dirichlet boundary conditions, but also improves the dependence on $d$ for the Poisson equation with inhomogeneous Dirichlet boundary conditions, as compared to previous quantum algorithms.

\begin{table}
\begin{center}
\notsotiny{
\renewcommand{\arraystretch}{1.25}
\begin{tabular}{|@{\hspace{1.5mm}}c|c|c|c|c|}
    \hline
     & \textbf{Algorithm} & \textbf{Equation} & \textbf{Boundary conditions}  & \textbf{Complexity} \\
    \hline
    \parbox[t]{1.5mm}{\multirow{5}{*}{\rotatebox[origin=c]{90}{Classical}}}
     & FDM/FEM/FVM  & general & general & $\poly((1/\epsilon)^d)$ \\
    \cline{2-5}
     & Adaptive FDM/FEM \cite{BG87}  & general & general & $\poly((\bm{\log(1/\epsilon)})^d)$ \\
    \cline{2-5}
     & Spectral method \cite{Ghe07,STW11} & general & general & $\poly((\bm{\log(1/\epsilon)})^d)$ \\
    \cline{2-5}
     & Sparse grid FDM/FEM \cite{BG04,Zen91}  & general & general & $\poly((1/\epsilon)(\log(1/\epsilon))^d)$ \\
    \cline{2-5}
     & Sparse grid spectral method \cite{SY10,SY12}  & elliptic & general & $\poly(\bm{\log(1/\epsilon)}(\log\log(1/\epsilon))^d)$ \\
    \hline
    \parbox[t]{1.5mm}{\multirow{7}{*}{\rotatebox[origin=c]{90}{Quantum}}}
     & FEM \cite{MP16}  & Poisson & homogeneous & $\poly(d, 1/\epsilon)$ \\
    \cline{2-5}
     & FDM \cite{cao2013quantum}  & Poisson & homogeneous Dirichlet & $\bm{d} \poly( \bm{\log d}, 1/\epsilon)$ \\
    \cline{2-5}
     & FDM \cite{CJO17}  & wave & homogeneous & $\bm{d^{5/2}} \poly(1/\epsilon)$ \\
    \cline{2-5}
     & FVM \cite{Fillion-Gourdeau2018}  & hyperbolic & periodic & $\bm d \poly(1/\epsilon)$ \\
    \cline{2-5}
     & Adaptive FDM [this paper]  & Poisson & periodic, homogeneous & $d^{13/2} \poly( \log d, \bm{\log(1/\epsilon)})$  \\
    \cline{2-5}
     & Spectral method [this paper]  & Poisson & homogeneous Dirichlet & $\bm{d \poly(\log d, \log(1/\epsilon))}$ \\
    \cline{2-5}
     & Spectral method [this paper]  & elliptic & inhomogeneous Dirichlet & $\bm{d^2 \poly(\log(1/\epsilon))}$ \\
    \hline
\end{tabular}
}
\end{center}
\caption{
Summary of the time complexities of classical and quantum algorithms for $d$-dimensional PDEs with error tolerance $\epsilon$. Portions of the complexity in bold represent best known dependence on that parameter.
\label{tab:alg-compare}
}
\end{table}


The remainder of the paper is structured as follows. \sec{pde-setup} introduces technical details about linear PDEs and formally states the problem we solve. \sec{fdm} covers our FDM algorithm for Poisson's equation. \sec{spectral} details the spectral algorithm for elliptic PDEs. Finally, \sec{discussion} concludes with a brief discussion of the results, their possible applications, and some open problems.

\section{Linear PDEs}
\label{sec:pde-setup}

In this paper, we focus on systems of linear PDEs. Such equations can be written in the form
\begin{equation}
\L(u(\bm{x}))=f(\bm{x}),
\label{eq:pde}
\end{equation}
where the variable $\bm{x}=(x_1,\ldots,x_d) \in \C^d$ is a $d$-dimensional vector, the solution $u(\bm{x})\in\C$ and the inhomogeneity $f(\bm{x})\in\C$ are scalar functions, and $\L$ is a linear differential operator acting on $u(\bm{x})$. In general, $\L$ can be written in a linear combination of $u(\bm{x})$ and its derivatives. A linear differential operator $\L$ of order $h$ has the form
\begin{equation}
\L(u(\bm{x}))=\sum_{\|\bm{j}\|_1\le h}A_{\bm{j}}(\bm{x})\frac{\partial^{\bm{j}}}{\partial \bm{x}^{\bm{j}}}u(\bm{x}),
\label{eq:operator}
\end{equation}
where $\bm{j}=(j_1,\ldots,j_d)$ is a $d$-dimensional non-negative vector with $\norm{\bm{j}}_1 = j_1+\dots+j_d\le h $, $A_j(\bm{x}) \in \C$, and
\begin{equation}
\frac{\partial^{\bm{j}}}{\partial \bm{x}^{\bm{j}}}u(\bm{x})=\frac{\partial^{j_1}}{\partial x_1^{j_1}}\cdots\frac{\partial^{j_d}}{\partial x_d^{j_d}}u(\bm{x}).
\label{eq:derivative}
\end{equation}
The problem reduces to a system of linear ordinary differential equations (ODEs) when $d=1$. For $d\ge2$, we call \eq{pde} a (multi-dimensional) PDE.

For example, systems of first-order linear PDEs can be written in the form
\begin{equation}
\sum_{j=1}^dA_j(\bm{x})\frac{\partial u(\bm{x})}{\partial x_j}+A_0(\bm{x})u(\bm{x})=f(\bm{x}),
\label{eq:fpde}
\end{equation}
where
$A_j(\bm{x}), A_0(\bm{x}), f(\bm{x}) \in\C$ for $j \in \range{d} := \{1,\ldots,d\}$.
Similarly, systems of second-order linear PDEs can be expressed in the form
\begin{equation}
\sum_{j_1,j_2=1}^dA_{j_1j_2}(\bm{x})\frac{\partial^2 u(\bm{x})}{\partial x_{j_1} \partial x_{j_2}}+\sum_{j=1}^dA_j(\bm{x})\frac{\partial u(\bm{x})}{\partial x_j}+A_0(\bm{x})u(\bm{x})=f(\bm{x}),
\label{eq:spde}
\end{equation}
where $A_{j_1,j_2}(\bm{x}), A_j(\bm{x}), A_0(\bm{x}), f(\bm{x}) \in\C$ for $j_1,j_2,j \in \range{d}$. A well-known second-order linear PDEs is the \emph{Poisson equation}
\begin{equation}
\Delta u(\bm{x}):=\sum_{j=1}^d\frac{\partial^2}{\partial x_j^2}u(\bm{x})=f(\bm{x}).
\label{eq:Poisson}
\end{equation}

A linear PDE of order $h$ is called \emph{elliptic} if its differential operator \eq{operator} satisfies
\begin{equation}
\sum_{\|\bm{j}\|_1= h}A_{\bm{j}}(\bm{x})\bm{\xi}^{\bm{j}}\ne0,
\label{eq:elliptic}
\end{equation}
for all nonzero $\bm{\xi}^{\bm{j}}=\xi_1^{j_1}\ldots\xi_d^{j_d}$ with $\xi_1,\ldots,\xi_d\in\R^m$ and all $\bm{x}$.
Note that ellipticity only depends on the highest-order terms.
When $h=2$, the linear PDE \eq{spde} is called a second-order elliptic PDE if and only if $A_{j_1j_2}(\bm{x})$ is positive-definite or negative-definite for any $\bm{x}$. In particular, the Poisson equation \eq{Poisson} is a second-order elliptic PDE.

We consider a class of
elliptic PDEs that also satisfy the condition
\begin{equation}
C := 1-\sum_{j_1=1}^d\frac{1}{|A_{j_1,j_1}(\bm{x})|}\sum_{j_2\in\range{d}\backslash\{j_1\}} |A_{j_1,j_2}(\bm{x})| > 0
\label{eq:DDM}
\end{equation}
for all $\bm{x}$.
We call this condition \emph{\gdd}, since it is a strengthening of the standard (strict) diagonal dominance condition 
\begin{equation}
d-\sum_{j_1=1}^d\frac{1}{|A_{j_1,j_1}(\bm{x})|}\sum_{j_2\in\range{d}\backslash\{j_1\}} |A_{j_1,j_2}(\bm{x})| > 0.
\end{equation}
Observe that \eq{DDM} holds for the Poisson equation \eq{Poisson} with $C=1$.

In this paper, we focus on the following boundary value problem:

\begin{problem}\label{prb:pde}
In the \emph{quantum PDE problem}, we are given a system of second-order elliptic equations
\begin{equation}
\L(u(\bm{x}))=\sum_{\|\bm{j}\|_1=2}A_{\bm{j}}\frac{\partial^{\bm{j}}}{\partial \bm{x}^{\bm{j}}}u(\bm{x})=\sum_{j_1,j_2=1}^dA_{j_1j_2}\frac{\partial^2 u(\bm{x})}{\partial x_{j_1} \partial x_{j_2}}=f(\bm{x})
\end{equation}
satisfying
the \gdd\ condition \eq{DDM},
where the variable $\bm{x}=(x_1,\ldots,x_d) \in \D=[-1, 1]^d$ is a $d$-dimensional vector, the inhomogeneity $f(\bm{x})\in\C$ is a scalar function of $\bm{x}$ satisfying $f(\bm{x}) \in C^\infty$, and the linear coefficients $A_{\bm{j}}\in\C$. We are also given boundary conditions $u(\bm{x}) = \gamma(\bm{x}) \in \partial\D$ or $\frac{\partial u(\bm{x})}{\partial x_j} {\big |}_{x_j= \pm 1} = \gamma (\bm{x}) |_{x_j = \pm 1} \in \partial\D$ where $\gamma (\bm{x}) \in C^{\infty} $. We assume there exists a weak solution $\hat{u}(\bm{x})\in\C$ for the boundary value problem (see Reference~\cite[Section 6.1.2]{eve10}). Given oracles that compute the coefficients $A_{\bm{j}}$, and that prepare normalized states $|\gamma(\bm{x})\rangle$  and $|f(\bm{x})\rangle$ whose amplitudes are proportional to $\gamma(\bm{x})$ and $f(\bm{x})$ on a set of interpolation nodes $\bm{x}$, the goal is to output a quantum state $|u(\bm{x})\rangle$ whose amplitudes are proportional to $u(\bm{x})$ on a set of interpolation nodes $\bm{x}$.
\end{problem}


\section{Finite difference method}
\label{sec:fdm}

We now describe our first approach to quantum algorithms for linear PDEs, based on the finite difference method (FDM). Using this approach, we show the following.


\begin{restatable}{theorem}{theoremFDM}\label{thm:main_fdm}
There exists a quantum algorithm that outputs a state $\epsilon$-close to $| u \rangle$ that runs in time
\begin{align}
\tilde O\biggl( d^{6.5} \log^{4.5} \Bigl(\Bigl| \frac{\d^{2k+1} u}{\d x^{2k+1}} \Bigr| / \epsilon \Bigr) \sqrt{ \log \Bigl[d^{4} \log^{3} \Bigl(\Bigl| \frac{\d^{2k+1} u}{\d x^{2k+1}} \Bigr| / \epsilon \Bigr)/ \epsilon\Bigr]} \biggr)
\end{align}
and makes
\begin{align}
\tilde O\biggl( d^{4}  \log^{3} \Bigl(\Bigl| \frac{\d^{2k+1} u}{\d x^{2k+1}} \Bigr| / \epsilon \Bigr) \sqrt{ \log \Bigl[d^{4} \log^{3} \Bigl(\Bigl| \frac{\d^{2k+1} u}{\d x^{2k+1}} \Bigr| / \epsilon \Bigr)/ \epsilon\Bigr]} \biggr)
\end{align}
queries to the oracle for $\vec f$.
\end{restatable}

To show this, we first construct a linear system corresponding to the finite difference approximation of Poisson's equation with periodic boundary conditions and bound the error of this high-order FDM in \sec{fdm_system} (\lem{high_order}). Then we bound the condition number of this system in \sec{fdm_cnum} (\lem{condnumpbc} and \lem{cnum_higherd}) and bound the error of approximation in \sec{fdm_error} (\lem{error_analysis}). We use these results to give an efficient quantum algorithm in \sec{fdm_alg}, establishing \thm{main_fdm}. We conclude by discussing how to use the method of images to apply this algorithm for Neumann and Dirichlet boundary conditions in \sec{fdm_images}.

The FDM approximates the derivative of a function $f$ at a point $\bm{x}$ in terms of the values of $f$ on a finite set of points near $\bm{x}$.
Generally there are no restrictions on where these points are located relative to $\bm{x}$, but they are typically taken to be uniformly spaced points with respect to a certain coordinate. This corresponds to discretizing $[-1,1]^d$ (or $[0,2 \pi )^d$) to a $d$-dimensional rectangular lattice (where we use periodic boundary conditions).

For a scalar field, in which $u(\bm{x}) \in \C$, the canonical elliptic PDE is Poisson's equation \eq{Poisson},
which we consider solving on $[0, 2\pi)^d$ with periodic boundary conditions. This also implies results for
the domain $\Omega = [-1,1]^d$ under Dirichlet ($u(\partial \Omega )=0$) and Neumann ($\hat n \cdot \nabla u(\partial \Omega) =0$ where $\hat n$ denotes the normal direction to $\partial \Omega$, which for domain $\Omega=[-1,1]^d$ is equivalent to $\frac{\partial u}{\partial x_j}\big|_{x_j=\pm 1}=0$ for $j \in [d]$) boundary conditions.

\subsection{Linear system}
\label{sec:fdm_system}

To approximate the second derivatives appearing in Poisson's equation, we apply the central finite difference formula of order $2k$. Taking $x_j=j h$ for a lattice with spacing $h$, this formula gives the approximation
\begin{align}
f''(0) \approx \frac{1}{h^2} \sum_{j=-k}^{k} r_j f(jh)
\end{align}
where the coefficients are \cite{kivlichan2017bounding,li2005general}
\begin{align}
  r_j:=
  \begin{cases}
    \frac{2(-1)^{j+1} (k!)^2}{j^2 (k-j)!(k+j)!} & j \in [k]\\
    -2 \sum_{j=1}^{k} r_j & j=0 \\
    r_{-j} & j \in -[k].
  \end{cases}
  \label{eq:fdm_coeffs}
\end{align}
We leave the dependence on $k$ implicit in this notation. The following lemma characterizes the error of this formula.

\begin{lemma}[{\cite[Theorem 7]{kivlichan2017bounding}}]\label{lem:high_order}
Let $k\ge 1$ and suppose $f(x) \in C^{2k+1}$ for $x \in \mathbb{R}$. Define the coefficients $r_j$ as in \eq{fdm_coeffs}. Then
\begin{align}
  \frac{\d^2 u(x_0) }{\d x^2} =
 \frac{1}{h^2} \sum_{j=-k}^{k} r_j f(x_0+jh) + O\Bigl( \Bigl|\frac{\d^{2k+1} u}{\d x^{2k+1}}\Bigr| \Bigl(\frac{eh}{2}\Bigr)^{2k-1}\Bigr)
\end{align}
where
\begin{align}
  \Bigl|\frac{\d^{2k+1} u}{\d x^{2k+1}}\Bigr| := \max_{y \in [x_0-kh,x_0+kh] } \Bigl|\frac{\d^{2k+1} u}{\d x^{2k+1}} (y)\Bigr|.
\end{align}
\end{lemma}

Since we assume periodic boundary conditions and apply the same FDM formula at each lattice site, the matrices we consider are circulant.
Define the $2n \times 2n$ matrix $S$ to have entries $S_{i,j} = \delta_{i,j+1 \bmod 2n} $.
If we represent the solution $u(x)$ as a vector $\vec u = \sum_{j=1}^{2n} u( \pi j / n) \vec e_j$, then we can approximate Poisson's equation using a central difference formula as
\begin{align}
\frac{1}{h^2 }L \vec u = \frac{1}{h^2} \Bigl(r_0 I + \sum_{j=1}^{k} r_j (S^j + S^{-j})\Bigr) \vec u = \vec f
\end{align}
where $\vec f = \sum_{j=1}^{2n} f( \pi j / n) \vec e_j$. The solution $\vec u$ corresponds exactly with the quantum state we want to produce, so we do not have to perform any post-processing such as in Reference~\cite{CJO17} and other quantum differential equation algorithms. The matrix in this linear system is just the finite difference matrix, so it suffices to bound its condition number and approximation error (whereas previous quantum algorithms involved more complicated linear systems).

\subsection{Condition number}
\label{sec:fdm_cnum}

The following lemma characterizes the condition number of a circulant Laplacian on $2n$ points.

\begin{lemma} \label{lem:condnumpbc} For $k < (6/\pi^2)^{1/3} n^{2/3}$, the matrix $L = r_0 I + \sum_{j=1}^{k} r_j (S^j + S^{-j})$ with $r_j$ as in \eq{fdm_coeffs} has condition number $\kappa (L) = O(n^2)$.
\end{lemma}
\begin{proof}

We first upper bound $\norm{L}$ using Gershgorin's circle theorem \cite{horn2012matrix} (a similar argument appears in Reference~\cite{kivlichan2017bounding}).
Note that
\begin{align}
  |r_j| = \frac{2(k!)^2}{j^2 (k-j)!(k+j)!} \leq \frac{2}{j^2}
  \label{eq:rbound}
\end{align}
since
\begin{align}
  \frac{(k!)^2}{(k-j)!(k+j)!} = \frac{k(k-1) \cdots (k-j+1)}{(k+j)(k+j-1) \cdots (k+1)}<1.
\end{align}
The radii of the Gershgorin discs are
\begin{align}
2 \sum_{j=1}^{k} |r_j| & \leq 2 \sum_{j=1}^{k} \frac{2}{j^2}  \leq \frac{2\pi^2}{3}.
\end{align}
The discs are centered at $r_0$, and
\begin{align}
|r_0| & \leq 2 \sum_{j=1}^{k} |r_j|  \leq \frac{2\pi^2}{3},
\end{align}
so $\norm{L} \leq \frac{4\pi^2}{3}$.

To lower bound $\norm{L^{-1}}$ we lower bound the (absolute value of the) smallest non-zero eigenvalue of $L$ (since by construction the all-ones vector is a zero eigenvector). Let $\omega := \exp ( \pi i /n)$. Since $L$ is circulant, its eigenvalues are
\begin{align}
\lambda_l  & = r_0 + \sum_{j=1}^{k} r_j (\omega^{lj} + \omega^{-lj})\\
& = r_0 + \sum_{j=1}^{k} 2 r_j \cos \Bigl(\frac{\pi l j }{n}\Bigr)\\
& = r_0 + \sum_{j=1}^{k} 2 r_j \left( 1 - \frac{\pi^2 l^2 j^2}{2n^2} + \frac{(\pi c_j)^4}{4!n^4} \cos \left(\frac{\pi c_j }{n}\right) \right) \\
& = \sum_{j=1}^{k} 2 r_j \left( - \frac{\pi^2 l^2 j^2}{2n^2} + \frac{(\pi c_j)^4}{4!n^4} \cos \left(\frac{\pi c_j }{n}\right) \right)
\end{align}
where the $c_j \in [0, lj]$  arise from the Taylor remainder theorem.
Using \eq{rbound}, we have
\begin{align}
\bigg|\lambda_1 + \frac{ \pi^2}{n^2} \sum_{j=1}^{k}  r_j j^2 \bigg|
&\le \frac{\pi^4k^3}{6n^4}.
\end{align}

We now compute the sum
\begin{align}
-\sum_{j=1}^{k}  r_j j^2  & = \sum_{j=1}^{k}  j^2  \frac{2(-1)^j (k!)^2}{j^2 (k+j)! (k-j)!} \\
& =2 (k!)^2 \sum_{j=1}^{k} \frac{(-1)^j }{ (k+j)! (k-j)!} \\
&= \frac{2 (k!)^2}{(2k)!} \sum_{j=1}^{k} (-1)^j \binom{2k}{k+j} \\
&= \frac{2 (k!)^2}{(2k)!} \sum_{j=k+1}^{2k} (-1)^{j+k} \binom{2k}{j} \\
& = (-1)^k \frac{(k!)^2}{(2k)!} \sum_{j=0,\, j \neq k}^{2k} (-1)^{j} \binom{2k}{j} \\
& =(-1)^k \frac{(k!)^2}{(2k)!} \biggl((1-1)^{2k}- (-1)^k \binom{2k}{k}\biggr) \\
& = -1.
\end{align}
Therefore, we have
\begin{align}
\lambda_1 &\le - \frac{ \pi^2}{n^2} + \frac{\pi^4 k^3}{6n^4}.
\end{align}

Finally, we see that
\begin{align}
  \kappa(L)
  &= \norm{L}\norm{L^{-1}} \\
  &\le \frac{4 \pi^2}{3}\Bigl(\frac{\pi^2}{n^2}-\frac{\pi^4 k^3}{6n^4}\Bigr)^{-1} \\
  &= \frac{4}{3}n^2\Bigl(1-\frac{\pi^2k^3}{6n^2} \Bigr)^{-1}
\end{align}
which is $O(n^2)$ provided $k < (6/\pi^2)^{1/3} n^{2/3}.$
\end{proof}

In $d$ dimensions, a similar analysis holds.

\begin{lemma} \label{lem:cnum_higherd} For $k < (6/\pi^2)^{1/3} n^{2/3}$, let $L := r_0 I + \sum_{j=1}^{k} r_j (S^j + S^{-j})$ with $r_j$ as in \eq{fdm_coeffs}. The matrix $L' := L \otimes I^{\otimes d-1} + I \otimes L \otimes I^{\otimes d-2} + \dots + I^{\otimes d-1} \otimes L$ has condition number $\kappa (L') = O(dn^2)$.
\end{lemma}

\begin{proof}
By the triangle inequality for spectral norms, $\norm{L'} \leq d \norm{L}$.
Since $L$ has zero-sum rows by construction, the all-ones vector lies in its kernel, and thus the smallest non-zero eigenvalue of $L$ is the same as that of $L'$. Therefore we have
\begin{align}
  \kappa(L')
  &\le \frac{4}{3}d n^2 \Bigl(1-\frac{\pi^2k^3}{6n^2} \Bigr)^{-1}
\end{align}
which is $O(dn^2)$ provided $k < (6/\pi^2)^{1/3} n^{2/3}.$
\end{proof}

\subsection{Error analysis}
\label{sec:fdm_error}

There are two types of error relevant to our analysis: the FDM error and the QLSA error. We assume that we are able to perfectly generate states proportional to $\vec f$.
The FDM errors arise from the remainder terms in the finite difference formulas and from inexact approximations of the eigenvalues.

We introduce several states for the purpose of error analysis.
Let $|u \rangle$ be the quantum state that is proportional to $\vec u = \sum_{j \in \mathbb{Z}_{2n}^d } u(\pi j/n) \bigotimes_{i=1}^d e_{j_i}$ for the exact solution of the differential equation.
Let $| \bar u \rangle$ be the state output by a QLSA that exactly solves the linear system.
Let $| \tilde u \rangle$ be the state output by a QLSA with error.
Then the total error of approximating $|u \rangle $ by $| \tilde u \rangle$ is bounded by
\begin{align}
\norm{| u \rangle -| \tilde u \rangle }
& \leq \norm{| u \rangle -| \bar  u \rangle } + \norm{| \bar u \rangle -| \tilde u \rangle } \\
& = \epsilon_{\mathrm{FDM}} + \epsilon_{\mathrm{QLSA}}
\end{align}
\noindent and without loss of generality we can take $\epsilon_{\mathrm{FDM}}$ and $\epsilon_{\mathrm{QLSA}}$ to be of the same order of magnitude.

\begin{lemma}\label{lem:error_analysis}
Let $u(\vec x)$ be the exact solution of $( \sum_{i=1}^d \frac{\d^2 }{\d x_i^2}) u (\vec x) = f(\vec x) $.
Let $\vec u \in \mathbb{R}^{(2n)^d}$ encode the exact solution in the sense that $\vec u = \sum_{j \in \mathbb{Z}_{2n}^d } u(\pi j/n) \bigotimes_{i=1}^d e_{j_i}$.
Let $\bar u \in \mathbb{R}^{(2n)^d}$ be the exact solution of the FDM linear system $\frac{1}{h^2} L' \bar u = \vec f$, where $L'$ is a $d$-dimensional $(2k)$th-order Laplacian as above with $k<(6/\pi^2)^{1/3}n^{3/2}$, and $\vec f = \sum_{j=1}^{2n} f( \pi j / n) \vec e_j$.
Then $\norm{\vec u - \bar u} \leq O( 2^{d/2} n^{(d/2)-2k+1} \bigl| \frac{\d^{2k+1} u}{\d x^{2k+1}} \bigr| (e^2/4)^k)$.
\end{lemma}

\begin{proof}
The remainder term of the central difference formula is $O( \bigl| \frac{\d^{2k+1} u}{\d x^{2k+1}} \bigr| h^{2k-1} (e/2)^{2k} )$, so
\begin{align}
\frac{1}{h^2} L' \vec u & = \vec f + O\Bigl( \Bigl| \frac{\d^{2k+1} u}{\d x^{2k+1}} \Bigr| (eh/2)^{2k-1}\Bigr)  \vec \epsilon
\end{align}
where $\vec \epsilon$ is a $(2n)^d$ dimensional vector whose entries are $O(1)$. This implies
\begin{align}
\frac{1}{h^2} L' (\vec u - \bar u) & =O\Bigl( \Bigl| \frac{\d^{2k+1} u}{\d x^{2k+1}} \Bigr| (eh/2)^{2k-1} \Bigr) \vec \epsilon
\end{align}
and therefore
\begin{align}
 \norm{\vec u - \bar u} & =O\Bigl( \Bigl| \frac{\d^{2k+1} u}{\d x^{2k+1}} \Bigr| (eh/2)^{2k+1} \Bigr) \norm{(L')^{-1} \vec \epsilon} \\
 & =O\Bigl( (2n)^{d/2}  \Bigl| \frac{\d^{2k+1} u}{\d x^{2k+1}} \Bigr| (eh/2)^{2k+1} / \lambda_1\Bigr).
\end{align}

By \lem{condnumpbc} we have $\lambda_1 = \Theta (1/n^2)$, and since $h = \Theta (1/n)$, we have
\begin{align}
 \norm{\vec u - \bar u}
 &=O \Bigl( 2^{d/2} n^{(d/2)-2k+1}  \Bigl| \frac{\d^{2k+1} u}{\d x^{2k+1}} \Bigr| (e/2)^{2k}\Bigr)
\end{align}
as claimed.
\end{proof}


\subsection{FDM algorithm}
\label{sec:fdm_alg}

To apply QLSAs, we must consider the complexity of simulating Hamiltonians that correspond to Laplacian FDM operators.
For periodic boundary conditions, the Laplacians are circulant, so they can be diagonalized by the QFT $F$ (or a tensor product of QFTs for the multi-dimensional Laplacian $L'$), i.e., $D=F^{\dagger}LF$ is diagonal. In this case the simplest way to simulate $\exp (i Lt)$ is to perform the inverse QFT, apply controlled phase rotations to implement $\exp (iDt)$, and perform the QFT.
Reference~\cite{schuch2003programmable} shows how to exactly implement arbitrary diagonal unitaries on $m$ qubits using $O(2^m)$ gates. Since we consider Laplacians on $n$ lattice sites, simulating $\exp (iLt)$ takes $O(n)$ gates with the dominant contribution coming from the phase rotations (alternatively, the methods of Reference~\cite{welch2014efficient} or Reference~\cite{berry2015simulating} could also be used).
Using this Hamiltonian simulation algorithm in a QLSA for the FDM linear system gives us the following theorem. We restate \thm{main_fdm} as follows.

\theoremFDM*

\begin{proof}
We use the Fourier series based QLSA from Reference~\cite{CKS15}.
By Theorem 3 of that work, the QLSA makes $O (\kappa \sqrt{ \log (\kappa / \epsilon_{\mathrm{QLSA}})} )$ uses of a Hamiltonian simulation algorithm and uses of the oracle for the inhomogeneity.
For Hamiltonian simulation we use $d$ parallel QFTs and phase rotations as described in Reference~\cite{schuch2003programmable}, for a total of $O (dn\kappa \sqrt{\log (\kappa / \epsilon_{\mathrm{QLSA}})} )$ gates. The condition number for the $d$-dimensional Laplacian scales as $\kappa = O(dn^2)$.

We take $\epsilon_{\mathrm{FDM}}$ and $\epsilon_{\mathrm{QLSA}}$ to be of the same order and just write $\epsilon$. Then the QLSA has time complexity $O( d^{2}n^{3} \sqrt{ \log (dn^2/ \epsilon)} )$ and query complexity $O(dn^2 \log (dn^2 / \epsilon ) )$.
The adjustable parameters are the number of lattice sites $n$ and the order $2k$ of the finite difference formula.
To keep the error below the target error of $\epsilon$ we require
\begin{align}
2^{d/2} n^{(d/2)-2k+1}  \Bigl| \frac{\d^{2k+1} u}{\d x^{2k+1}} \Bigr| (e/2)^{2k} & = O( \epsilon ),
\end{align}
or equivalently,
\begin{align}
(-d/2) + (2k-1-(d/2)) \log (n) -2k \log (e/2) & = \Omega \Bigl( \log \Bigl(\Bigl| \frac{\d^{2k+1} u}{\d x^{2k+1}} \Bigr| / \epsilon  \Bigr)\Bigr).
\label{eq:control_error}
\end{align}
Now we focus on the choice of adjustable $n$ and $k$ relying on $\epsilon$. This procedure is inspired by the classical adaptive FDM \cite{BG87}, so we call it the adaptive FDM approach.
We must have $2k-1>d/2$ for the left-hand side of \eq{control_error} to be positive for large $n$. Indeed, we find the best performance by taking $k$ as large as possible subject to the assumption of \lem{condnumpbc}, i.e.,
$k = cn^{2/3}$ where $c:=(6/\pi^2)^{1/3}$.
For this choice of $k$ and for $n$ sufficiently large, \eq{control_error} is equivalent to
\begin{align}
k\log (n) = cn^{2/3} \log (n)  = \Omega \Bigl( \log \Bigl(\Bigl| \frac{\d^{2k+1} u}{\d x^{2k+1}} \Bigr| / \epsilon  \Bigr)\Bigr).
\label{eq:kn2}
\end{align}
To satisfy the condition $2cn^{2/3}-1 > d/2$, we must have $n = \Omega(d^{3/2})$. Combining this observation with \eq{kn2}, we choose
\begin{align}
n = \Theta \Bigl( d^{3/2} \log^{3/2} \Bigl(\Bigl| \frac{\d^{2k+1} u}{\d x^{2k+1}} \Bigr| / \epsilon  \Bigr) \Bigr)
\end{align}
so that
\begin{align}
k = cn^{2/3} = \Theta \Bigl( d\log \Bigl(\Bigl| \frac{\d^{2k+1} u}{\d x^{2k+1}} \Bigr| / \epsilon  \Bigr) \Bigr).
\end{align}
The QLSA
then has the stated time complexity
\begin{align}
\tilde O( d^{2}n^{3} \sqrt{ \log (dn^2/ \epsilon)} ) = O\biggl( d^{6.5} \log^{4.5} \Bigl(\Bigl| \frac{\d^{2k+1} u}{\d x^{2k+1}} \Bigr| / \epsilon  ) \sqrt{ \log \Bigl[d^{4} \log^{3} \Bigl(\Bigl| \frac{\d^{2k+1} u}{\d x^{2k+1}} \Bigr| / \epsilon \Bigr)/ \epsilon\Bigr]} \biggr),
\end{align}
and makes
\begin{align}
\tilde O(dn^2 \log (dn^2 / \epsilon ) ) = O\biggl( d^{4 }  \log^{3} \Bigl(\Bigl| \frac{\d^{2k+1} u}{\d x^{2k+1}} \Bigr| / \epsilon \Bigr) \sqrt{ \log \Bigl[d^{4} \log^{3} \Bigl(\Bigl| \frac{\d^{2k+1} u}{\d x^{2k+1}} \Bigr| / \epsilon )/ \epsilon\Bigr]} \biggr).
\end{align}
queries to the oracle for $\vec f$.
\end{proof}

This can be compared to the cost of using the conjugate gradient method to solve the same linear system classically.
The sparse conjugate gradient algorithm for an $N \times N$ matrix has time complexity $O(Ns \sqrt{\kappa} \log (1 / \epsilon))$. For arbitrary dimension $N=\Theta (n^d)$, we have $s=dk=cdn^{2/3}$ and $\kappa = O(dn^2)$, so that the time complexity is $ O( d^{4+3d/2} \log (1/\epsilon) \allowbreak \log^{5/2+3d/2} ( \bigl|\frac{\d^{2k+1} u}{\d x^{2k+1}} \bigr| / \epsilon ) )$.
Alternatively, $d$ fast Fourier transforms could be used, although this will generally take $\Omega (n^d) = \Omega ( d^{3d/2 } \log^{3d/2 } ( \bigl|\frac{\d^{2k+1} u}{\d x^{2k+1}} \bigr| / \epsilon )) $ time.


\subsection{Boundary conditions via the method of images}
\label{sec:fdm_images}

We can apply the method of images to deal with homogeneous Neumann and Dirichlet boundary conditions using the algorithm for periodic boundary conditions described above. In the method of images, the domain $[-1,1]$ is extended to include all of $\mathbb{R}$, and the boundary conditions are related to symmetries of the solutions.
For a pair of Dirichlet boundary conditions there are two symmetries: the solutions are anti-symmetric about $-1$ (i.e., $f(-x-1)=-f(x-1)$) and anti-symmetric about 1 (i.e., $f(1+x)= - f(1-x)$). Continuity and anti-symmetry about $-1$ and $1$ imply $f(-1)=f(1)=0$, and furthermore that $f(x)=0$ for all odd $x \in \mathbb{Z}$ and that $f(x+4)=f(x)$ for all $x \in \R$. For Neumann boundary conditions, the solutions are instead symmetric about $-1$ and $1$, which also implies $f(x+2)=f(x)$ for all $x \in \R$.

We would like to combine the method of images with the FDM to arrive at finite difference formulas for this special case.
In both cases, the method of images implies that the solutions are periodic, so without loss of generality we can consider a lattice on $[0, 2 \pi)$ instead of a lattice on $\mathbb{R}$.
It is useful to think of this lattice in terms of the cycle graph on $2n$ vertices, i.e., $(V,E)= ( \mathbb{Z}_{2n} , \{(i,i+1) \mid i \in \mathbb{Z}_{2n} \} )$, which means that the vectors encoding the solution $u(x)$ will lie in $\mathbb{R}^{2n}$. Let each vector $\vec e_j$ correspond to the vertex $j$. Then we divide $\mathbb{R}^{2n}$ into a symmetric and an anti-symmetric subspace, namely $\mathrm{span} \{ e_j + e_{2n+1-j} \}_{j=1}^{n} $ and $\mathrm{span} \{ e_j - e_{2n+1-j} \}_{j=1}^{n} $, respectively. Vectors lying in the symmetric subspace correspond to solutions that are symmetric about $0$ and $\pi$, so they obey Neumann boundary conditions at $0$ and $\pi$; similarly, vectors in the anti-symmetric space correspond to solutions obeying Dirichlet boundary conditions at $0$ and $\pi$.

Restricting to a subspace of vectors reduces the size of the FDM vectors and matrices we consider, and the symmetry of that subspace indicates how to adjust the coefficients.

If the FDM linear system is $L'' \vec u'' = \vec f''$ then $L''$ has entries
\begin{align}
L''_{i,j} =
\begin{cases}
r_{|i-j|} \pm r_{i+j-1} & i \leq k \\
r_{|i-j|} & k < i \leq n-k \\
r_{|i-j|} \pm r_{2n-i-j+1 } & n-k \leq i
\end{cases}
\end{align}
where $+$ ($-$) is chosen for Neumann (Dirichlet) boundary conditions and due to the truncation order $k$, $r_j=0$ for any $j > k$.
This is similar to how Laplacian coefficients are modified when imposing boundary conditions in discrete variable representations \cite{colbert1992novel}.

For the purpose of solving the new linear systems using quantum algorithms, we still treat these cases as obeying periodic boundary conditions. We assume access to an oracle that produces states $| f'' \rangle $ proportional to the inhomogeneity $f''(x)$. Then we apply the QLSA for periodic boundary conditions using $| f '' \rangle | \pm \rangle$ to encode the inhomogeneity, which will output solutions of the form $|u'' \rangle | \pm \rangle$. Here the ancillary state is chosen to be $| + \rangle$ ($| - \rangle$) for Neumann (Dirichlet) boundary conditions.

Typically, the (second-order) graph Laplacian for the path graph with Dirichlet boundary conditions has diagonal entries that are all equal to 2; however, using the above specification for the entries of $L$ leads to the $(1,1)$ and $(n,n)$ entries being 3 while the rest of the diagonal entries are 2.

To reproduce this case, we consider an alternative subspace restriction used in Reference~\cite{Spielman} to diagonalize the Dirichlet graph Laplacian. In this case it is easiest to consider the lattice of a cycle graph on $2n+2$ vertices, where the vertices $0$ and $n+1$ are selected as boundary points where the field takes the value 0. The relevant antisymmetric subspace is now $\mathrm{span}(\{ e_j - e_{2n+2-j} \}_{j=1}^{n})$ (which has no support on $e_0$ and $e_{n+1}$).

If we again write the linear system as $L'' \vec u'' = \vec f''$, then the Laplacian has entries

\[
L''_{i,j} =
\begin{cases}
r_{|i-j|} - r_{i+j} & i \leq k \\
r_{|i-j|} & k < i \leq n-k \\
r_{|i-j|} - r_{2n-i-j+2 } & n-k \leq i.
\end{cases}
\]

We again assume access to an oracle producing states proportional to $f''(x)$; however, we assume that this oracle operates in a Hilbert space with one additional dimension compared to the previous approaches (i.e., whereas previously we considered implementing $U$, here we consider implementing $\Bigl(\begin{smallmatrix}
U & \vec 0 \\ \vec 0^T & 1
\end{smallmatrix}\Bigr)$). With this oracle we again prepare the state $|f'' \rangle |- \rangle$ and solve Poisson's equation for periodic boundary conditions to output a state $|u'' \rangle |- \rangle$ (where $| u'' \rangle$ lies in an ($n+1$)-dimensional Hilbert space but has no support on the $(n+1)$st basis state).


\section{Multi-dimensional spectral method}
\label{sec:spectral}

We now turn our attention to the spectral method for multi-dimensional PDEs.
Since interpolation facilitates constructing a straightforward linear system, we develop a quantum algorithm based on the \emph{pseudo-spectral method} \cite{Ghe07,STW11,Tang06} for second-order elliptic equations with global strict diagonal dominance, under various boundary conditions. Using this approach, we show the following.

\begin{restatable}{theorem}{theoremSpectral}\label{thm:main}
Consider an instance of the quantum PDE problem as defined in \prb{pde} with Dirichlet boundary conditions \eq{dbc}. Then there exists a quantum algorithm that produces a state in the form of \eq{u_expand} whose amplitudes are proportional to $u(\bm{x})$ on a set of interpolation nodes $\bm{x}$ (with respect to the uniform grid nodes for periodic boundary conditions or the Chebyshev-Gauss-Lobatto quadrature nodes for non-periodic boundary conditions, as defined in in \eq{interpolation_nodes}), where $u(\bm{x})/\| u(\bm{x})\|$ is $\epsilon$-close to $\hat{u}(\bm{x})/\| \hat{u}(\bm{x})\|$ in $l_2$ norm for all nodes $\bm{x}$, succeeding with probability $\Omega(1)$, with a flag indicating success, using
\begin{equation}
  \biggl(\frac{d\|A\|_\Sigma}{C\|A\|_{\ast}}+qd^2\biggr)\poly(\log(g'/g\epsilon))
\end{equation}
queries to oracles as defined in \sec{state}.
Here $\|A\|_\Sigma:=\sum_{\|\bm{j}\|_1\le h}\|A_{\bm{j}}\|$, $\|A\|_{\ast}:=\sum_{j=1}^d |A_{j,j}|$, $C>0$ is defined in \eq{DDM}, and
\begin{align}
g = \min_{\bm{x}} \|\hat{u}(\bm{x})\|, \qquad g':=\max_{\bm{x}} \max_{n \in \N}\|\hat{u}^{(n+1)}(\bm{x})\|,\qquad \label{eq:main_thm_g} \\
q = \sqrt{\frac{\sum_{\|\bm{k}\|_{\infty}\le n} \sum_{j=1}^d \hat{f}^2_{\bm{k}}+(A_{j,j}\hat{\gamma}^{j+}_{\bm{k}})^2+(A_{j,j}\hat{\gamma}^{j-}_{\bm{k}})^2}{\sum_{\|\bm{k}\|_{\infty}\le n} \sum_{j=1}^d (\hat{f}_{\bm{k}}+A_{j,j}\hat{\gamma}^{j+}_{\bm{k}}+A_{j,j}\hat{\gamma}^{j-}_{\bm{k}})^2}}
.
\label{eq:main_thm_parameters}
\end{align}
The gate complexity is larger than the query complexity by a factor of $\poly(\log(d\|A\|_\Sigma/\epsilon))$.
\end{restatable}

After introducing the method, we discuss the complexity of the quantum shifted Fourier transform (\lem{qsft}) and the quantum cosine transform (\lem{qct}) in \sec{transform}.
These transforms are used as subroutines in our algorithm. Then we construct a linear system whose solution encodes the solution of the PDE in \sec{solver} (with a simple illustrative example presented in \app{Poisson}), analyze its condition number in \sec{condition} (\lem{condition}, established using \lem{svd_Fourier}, \lem{svd_Chebyshev}, and \lem{condition_Poisson}),
and consider the complexity of state preparation in \sec{state} (\lem{preparation}).
Finally, we prove our main result (\thm{main}) in \sec{main}.

In the spectral approach, we approximate the exact solution $\hat u(\bm{x})$ by a linear combination of basis functions
\begin{equation}
u(\bm{x}) = \sum_{\|\bm{k}\|_{\infty}\le n}c_{\bm{k}}\phi_{\bm{k}}(\bm{x})
\label{eq:basis_expand}
\end{equation}
for some $n \in \Z^+$. Here $\bm{k}=(k_1,\ldots,k_d)$ with $k_j\in \rangez{n+1}:=\{0,1,\ldots,n\}$, $c_{\bm{k}}\in \C$, and
\begin{equation}
\phi_{\bm{k}}(\bm{x}) = \prod_{j=1}^d\phi_{k_j}(x_j), \quad j\in\range{d}.
\label{eq:basis_tensor}
\end{equation}

We choose different basis functions for the case of periodic boundary conditions and for the more general case of non-periodic boundary conditions. When the boundary conditions are periodic, the algorithm implementation is more straightforward, and in some cases (e.g., for the Poisson equation), can be faster. Specifically, for any $k_j \in \rangez{n+1}$ and $x_j \in [-1,1]$, we take
\begin{equation}
\phi_{k_j}(x_j) = \begin{cases}
e^{i(k_j-\left\lfloor n/2 \right\rfloor)\pi x_j}, & \text{periodic conditions},\\
T_{k_j}(x_j) := \cos(k_j\arccos x_j), & \text{non-periodic conditions}.
\end{cases}
\label{eq:basis_function}
\end{equation}
Here $T_k$ is the degree-$k$ Chebyshev polynomial of the first kind.

The coefficients $c_{\bm{k}}$ are determined by demanding that $u(\bm{x})$ satisfies the ODE and boundary conditions at a set of \emph{interpolation nodes} $\{\bm{\chi}_{\bm{l}}=(\chi_{l_1}, \ldots, \chi_{l_d})\}_{\|\bm{l}\|_{\infty\le n}}$ with $l_j\in \rangez{n+1}$, where
\begin{equation}
\chi_{l_j} = \begin{cases}
\frac{2 l_j}{n+1}-1, & \text{periodic conditions},\\
\cos\frac{\pi l_j}{n}, & \text{non-periodic conditions}.
\end{cases}
\label{eq:interpolation_nodes}
\end{equation}
Here $\{\frac{2 l}{n+1}-1: l \in \rangez{n+1}\}$ are called the \emph{uniform grid nodes}, and $\{\cos\frac{\pi l}{n} : l \in \rangez{n+1}\}$ are called the \emph{Chebyshev-Gauss-Lobatto quadrature nodes}.

We require the numerical solution $u(\bm{x})$ to satisfy
\begin{equation}
\L(u(\bm{\chi}_{\bm{l}}))=f(\bm{\chi}_{\bm{l}}), \quad \forall \, l_j\in\rangez{n+1},~ j\in\range{d}.
\label{eq:spectral_pde}
\end{equation}
We would like to be able to increase the accuracy of the approximation by increasing $n$, so that
\begin{equation}
\|\hat{u}(\bm{x})-u(\bm{x})\|\rightarrow0 \quad \text{as} \quad n\rightarrow\infty.
\end{equation}

The convergence behavior of the spectral method is related to the smoothness of the solution. For a solution in $C^{r+1}$, the spectral method approximates the solution with $n=\poly({1}/{\epsilon})$. Furthermore, if the solution is in $C^{\infty}$, the spectral method approximates the solution to within $\epsilon$ using only $n=\poly(\log(1/\epsilon))$ \cite{STW11}. Since we require $k_j\in\rangez{n+1}$ for all $j\in\range{d}$, we have $(n+1)^d$ terms in total. Consequently, a classical pseudo-spectral method solves multi-dimensional PDEs with complexity $\poly(\log^d(1/\epsilon))$. Such classical spectral methods rapidly become infeasible since the number of coefficients $(n+1)^d$ grows exponentially with $d$.

Here we develop a quantum algorithm for multi-dimensional PDEs. The algorithm applies techniques from the quantum spectral method for ODEs \cite{CL19}. However, in the case of PDEs, the linear system to be solved is non-sparse. We address this difficulty using a quantum transform that restores sparsity.


\subsection{Quantum shifted Fourier transform and quantum cosine transform}
\label{sec:transform}

The well-known \emph{quantum Fourier transform} (QFT) can be regarded as an analogue of the discrete Fourier transform (DFT) acting on the amplitudes of a quantum state. The QFT maps the $(n+1)$-dimensional quantum state $v=(v_0,v_1,\ldots,v_n)\in\C^{n+1}$ to the state $\hat{v}=(\hat{v}_0,\hat{v}_1,\ldots,\hat{v}_n)\in\C^{n+1}$ with
\begin{equation}
\hat{v}_l = \frac{1}{\sqrt{n+1}}\sum_{k=0}^{n}\exp\Bigl(\frac{2\pi ikl}{n+1}\Bigr)v_k,\quad l\in\rangez{n+1}.
\label{eq:qft_rule}
\end{equation}
In other words, the QFT is the unitary transform
\begin{equation}
F_n:=\frac{1}{\sqrt{n+1}}\sum_{k,l=0}^n \exp\Bigl(\frac{2\pi ikl}{n+1}\Bigr)|l\rangle\langle k|.
\label{eq:qft_matrix}
\end{equation}

Here we also consider the \emph{quantum shifted Fourier transform} (QSFT), an analogue of the classical shifted discrete Fourier transform, which maps $v \in \C^{n+1}$ to $\hat v \in \C^{n+1}$ with
\begin{equation}
\hat{v}_l = \frac{1}{\sqrt{n+1}}\sum_{k=0}^{n}\exp\Bigl(\frac{2\pi i(k-\left\lfloor n/2 \right\rfloor)(l-(n+1)/2)}{n+1}\Bigr)v_k,\quad l\in\rangez{n+1}.
\label{eq:qsft_rule}
\end{equation}
In other words, the QSFT is the unitary transform
\begin{equation}
F^s_n:=\frac{1}{\sqrt{n+1}}\sum_{k,l=0}^n \exp\Bigl(\frac{2\pi i(k-\left\lfloor n/2 \right\rfloor)(l-(n+1)/2)}{n+1}\Bigr)|l\rangle\langle k|.
\label{eq:qsft_matrix}
\end{equation}

We define the multi-dimensional QSFT by the tensor product, namely
\begin{equation}
{\bm{F}}^s_n:=\frac{1}{\sqrt{(n+1)^d}}\sum_{\|\bm{k}\|_{\infty},\|\bm{l}\|_{\infty}\le n} \prod_{j=1}^d\exp\bigl(\tfrac{2\pi i(k_j-\left\lfloor n/2 \right\rfloor)(l_j-(n+1)/2)}{n+1}\bigr)|l_1\rangle\ldots|l_d\rangle\langle k_1|\ldots\langle k_d|,
\label{eq:mqsft_matrix}
\end{equation}
where $\bm{k}=(k_1,\ldots,k_d)$ and $\bm{l}=(l_1,\ldots,l_d)$ are $d$-dimensional vectors with $k_j, l_j\in \rangez{n}$.

The QSFT can be efficiently implemented as follows:

\begin{lemma}\label{lem:qsft}
The QSFT $F_n^s$ defined by \eq{qsft_matrix} can be performed with gate complexity $O(\log n \allowbreak \log\log n)$. More generally, the $d$-dimensional QSFT ${\bm F}_n^s$ defined by \eq{mqsft_matrix} can be performed with gate complexity $O(d \log n\log\log n)$.
\end{lemma}

\begin{proof}
The unitary matrix $F^s_n$ can be written as the product of three unitary matrices
\begin{equation}
F^s_n=S_nF_nR_n,
\end{equation}
where
\begin{equation}
R_n=\sum_{k=0}^n \exp\Bigl(-\frac{2\pi i k(n+1)/2}{n+1}\Bigr)|k\rangle\langle k|
\end{equation}
and
\begin{equation}
S_n=\sum_{l=0}^n \exp\Bigl(-\frac{2\pi i \left\lfloor n/2 \right\rfloor (l-(n+1)/2)}{n+1}\Bigr)|l\rangle\langle l|.
\end{equation}
It is well known that $F_n$ can be implemented with gate complexity $O(\log n\log\log n)$, and it is straightforward to implement $R_n$ and $S_n$ with gate complexity $O(\log n)$. Thus the total complexity is $O(\log n\log\log n)$.

We rewrite $\bm{v}$ in the form
\begin{equation}
\bm{v}=\sum_{\|\bm{k}\|_{\infty}\le n}\bm{v}_{\bm{k}}|k_1\rangle\ldots|k_d\rangle,
\end{equation}
where $\bm{v}_{\bm{k}}\in\C$ with $\bm{k}=(k_1,\ldots,k_d)$, and each $k_j\in \rangez{n}$ for $j\in \range{d}$.
The unitary matrix ${\bm{F}}^s_n$ can be written as the tensor product
\begin{equation}
{\bm{F}}^s_n=\bigotimes_{j=1}^d F^s_n.
\end{equation}
Performing the multi-dimensional QSFT is equivalent to performing the one-dimensional QSFT on each register. Thus, the gate complexity of performing ${\bm{F}}^s_n$ is $O(d\log n\log\log n)$.
\end{proof}

Another efficient quantum transformation is the \emph{quantum cosine transform} (QCT) \cite{KR01,RPB99}. The QCT can be regarded as an analogue of the discrete cosine transform (DCT). The QCT maps $v\in\C^{n+1}$ to $\hat{v}\in\C^{n+1}$  with
\begin{equation}
\hat{v}_l = \sqrt{\frac{2}{n}}\sum_{k=0}^{n} \delta_k\delta_l\cos\frac{kl\pi}{n}v_k,\quad l\in\rangez{n+1},
\label{eq:qct_rule}
\end{equation}
where
\begin{equation}
\delta_l := \begin{cases}
\frac{1}{\sqrt{2}} & l=0,n\\
1 & l\in\range{n-1}.
\end{cases}
\end{equation}
In other words, the QCT is the orthogonal transform
\begin{equation}
C_n:=\sqrt{\frac{2}{n}}\sum_{k,l=0}^{n} \delta_l\delta_k\cos\frac{kl\pi}{n}|l\rangle\langle k|.
\label{eq:qct_matrix}
\end{equation}

Again we define the multi-dimensional QCT by the tensor product, namely
\begin{equation}
{\bm{C}}_n:=\sqrt{\Bigl(\frac{2}{n}\Bigr)^d}\sum_{\|\bm{k}\|_{\infty},\|\bm{l}\|_{\infty}\le n} \prod_{j=1}^d\delta_{k_j}\delta_{l_j}\cos\frac{k_jl_j\pi}{n}|l_1\rangle\ldots|l_d\rangle\langle k_1|\ldots\langle k_d|,
\label{eq:mqct_matrix}
\end{equation}
where $\bm{k}=(k_1,\ldots,k_d)$ and $\bm{l}=(l_1,\ldots,l_d)$ are $d$-dimensional vectors with $k_j, l_j\in \rangez{n+1}$.

The classical DCT on $(n+1)$-dimensional vectors takes $\Theta(n\log n)$ gates, while the QCT on $(n+1)$-dimensional quantum states can be implemented with complexity $\poly(\log n)$. According to Theorem 1 of Reference~\cite{KR01}, the gate complexity of performing $C_n$ is $O(\log^2 n)$. We observe that this can be improved as follows.

\begin{lemma}\label{lem:qct}
The quantum cosine transform $C_n$ defined by \eq{qct_matrix} can be performed with gate complexity $O(\log n\log\log n)$. More generally, the multi-dimensional QCT $\bm{C}_n$ defined by \eq{mqct_matrix} can be performed with gate complexity $O(d\log n\log\log n)$.
\end{lemma}

\begin{proof}
According to the quantum circuit in Figure 2 of Reference~\cite{KR01}, $C_n$ can be decomposed into a QFT $F_{n+1}$, a permutation
\begin{equation}
P_n=
  \begin{pmatrix}
     &  &  &  & 1 \\
    1 &  &  &  & \\
     & 1 &  &  & \\
     &  & \ddots  & & \\
     &  &  &  1 & \\
  \end{pmatrix},
\end{equation}
and additional operations with $O(1)$ cost. The QFT $F_{n+1}$ has gate complexity $O(\log n \allowbreak \log\log n)$. We then consider an alternative way to implement $P_n$ that improves over the approach in \cite{PRB99}.

The permutation $P_n$ can be decomposed as
\begin{equation}
P_n=F_nT_nF_n^{-1},
\end{equation}
where $F_n$ is the Fourier transform \eq{qft_matrix} and $T_n = \sum_{k=0}^n e^{-\frac{2\pi i k}{n+1}} |k\rangle\langle k|$ is diagonal. The gate complexities of performing $F_n$ and $T_n$ are $O(\log n\log\log n)$ and $O(\log n)$, respectively. It follows that $C_n$ can be implemented with circuit complexity $O(\log n\log\log n)$.

The matrix ${\bm{C}}_n$ can be written as the tensor product
\begin{equation}
{\bm{C}}_n=\bigotimes_{j=1}^d C_n.
\end{equation}
As in \lem{qsft}, performing the multi-dimensional QCT is equivalent to performing a QCT on each register. Thus, the gate complexity of performing ${\bm{C}}_n$ is $O(d\log n\log\log n)$.
\end{proof}


\subsection{Linear system}
\label{sec:solver}

In this section we introduce the quantum PDE solver for the problem \eq{pde}. We construct a linear system that encodes the solution of \eq{pde} according to the pseudo-spectral method introduced above, using the QSFT/QCT introduced in \sec{transform} to ensure sparsity.

We consider a linear PDE problem (\prb{pde}) with periodic boundary conditions
\begin{equation}
u(\bm{x}+2\bm{v}) = u(\bm{x}) \quad \forall\, \bm{x} \in \D, ~ \forall\, \bm{v} \in \Z^d
\label{eq:pbc}
\end{equation}
or non-periodic Dirichlet boundary conditions
\begin{equation}
u(\bm{x}) = \gamma(\bm{x}) \quad \forall\, \bm{x} \in \partial\D.
\label{eq:dbc}
\end{equation}
According to the elliptic regularity theorem (Theorem 6 in Section 6.3 of Reference~\cite{eve10}), there exists a unique solution $\hat{u}(\bm{x})$ in $C^\infty$ for \prb{pde}.

We now show how to apply the Fourier and Chebyshev pseudo-spectral methods to this problem. Our goal is to obtain the quantum state
\begin{equation}
  |u\rangle
\propto \sum_{ \|\bm{k}\|_{\infty}, \|\bm{l}\|_{\infty}\le n} c_{\bm{k}}\phi_{\bm{k}}(\bm{\chi}_{\bm{l}})|l_1\rangle\ldots|l_d\rangle,
\label{eq:u_expand}
\end{equation}
where $\phi_{\bm{k}}(\bm{\chi}_{\bm{l}})$ is defined by \eq{basis_tensor} using \eq{basis_function} for the appropriate boundary conditions (periodic or non-periodic). This state corresponds to a truncated Fourier/Chebyshev approximation and is $\epsilon$-close to the exact solution $\hat{u}(\bm{\chi}_{\bm{l}})$ with $n=\poly(\log(1/\epsilon))$ \cite{STW11}. Note that this state encodes the values of the solution at the interpolation nodes \eq{interpolation_nodes} appropriate to the boundary conditions (the uniform grid nodes in the Fourier approach, for periodic boundary conditions, and the Chebyshev-Gauss-Lobatto quadrature nodes in the Chebyshev approach, for non-periodic boundary conditions).

Instead of developing our algorithm for the standard basis, we aim to produce a state
\begin{equation}
|c\rangle \propto \sum_{\|\bm{k}\|_{\infty}\le n} c_{\bm{k}}|k_1\rangle\ldots|k_d\rangle
\label{eq:u_coeff}
\end{equation}
that is the inverse QSFT/QCT of $|u\rangle$. We then apply the QSFT/QCT to transform back into the interpolation node basis.

The truncated spectral series of the inhomogeneity $f(\bm{x})$ and the boundary conditions $\gamma(\bm{x})$ can be expressed as
\begin{equation}
f(\bm{x}) = \sum_{\|\bm{k}\|_{\infty}\le n} \hat{f}_{\bm{k}}\phi_{\bm{k}}(\bm{x})
\label{eq:f_expand}
\end{equation}
and
\begin{equation}
\gamma(\bm{x}) = \sum_{\|\bm{k}\|_{\infty}\le n} \hat{\gamma}_{\bm{k}}\phi_{\bm{k}}(\bm{x}),
\label{eq:g_expand}
\end{equation}
respectively. We define quantum states $|f\rangle$ and $|\gamma\rangle$ by interpolating the nodes $\{\bm{\chi}_{\bm{l}}\}$ defined by \eq{interpolation_nodes} as
\begin{equation}
|f\rangle
\propto \sum_{\|\bm{k}\|_{\infty},\|\bm{l}\|_{\infty}\le n} \phi_{k_j}(\bm{\chi}_{\bm{l}})\hat{f}_{\bm{k}}|l_1\rangle\ldots|l_d\rangle,
\label{eq:f_inter}
\end{equation}
and
\begin{equation}
  |\gamma\rangle
\propto \sum_{\|\bm{k}\|_{\infty},\|\bm{l}\|_{\infty}\le n} \phi_{k_j}(\bm{\chi}_{\bm{l}})\hat{\gamma}_{\bm{k}}|l_1\rangle\ldots|l_d\rangle,
\label{eq:g_inter}
\end{equation}
respectively. These are the states that we assume we can produce using oracles. We perform the multi-dimensional inverse QSFT/QCT to obtain the states
\begin{equation}
|\hat{f}\rangle \propto \sum_{\|\bm{k}\|_{\infty}\le n} \hat{f}_{\bm{k}}|k_1\rangle\ldots|k_d\rangle,
\label{eq:f_coeff}
\end{equation}
and
\begin{equation}
|\hat{\gamma}\rangle \propto \sum_{\|\bm{k}\|_{\infty}\le n} \hat{\gamma}_{\bm{k}}|k_1\rangle\ldots|k_d\rangle.
\label{eq:g_coeff}
\end{equation}

Having defined these states, we now detail the construction of the linear system. At a high level, we construct two linear systems: one system $A \vec x = \vec f$ (where $\vec x$ corresponds to \eq{u_coeff}) describes the differential equation, and another system $B \vec x = \vec g$ describes the boundary conditions. We combine these into a linear system with the form
\begin{equation}
L \vec x =(A+B) \vec x = \vec f + \vec g.
\label{eq:linear_system}
\end{equation}
Even though we do not impose the two linear systems separately, we show that there exists a unique solution of \eq{linear_system} (which is therefore the solution of the simultaneous equations $A\vec x=\vec f$ and $B\vec x=\vec g$), since we show that $L$ has full rank, and indeed we upper bound its condition number in \sec{condition}.

Part of this linear system will correspond to just the differential equation
\begin{equation}
\L(u(\bm{\chi}_{\bm{l}}))=\sum_{\|\bm{j}\|_{1}= 2}A_{\bm{j}}\frac{\partial^{\bm{j}}}{\partial \bm{x}^{\bm{j}}}u(\bm{\chi}_{\bm{l}})=f(\bm{\chi}_{\bm{l}}),
\label{eq:inter_eq}
\end{equation}
while another part will come from imposing the boundary conditions on $\partial\D=\bigcup_{j \in [d]} \partial\D_j$, where $\partial\D_j:=\{\bm{x} \in \D \mid x_j=\pm1\}$ is a $(d-1)$-dimensional subspace.
More specifically, the boundary conditions
\begin{equation}
 u(\bm{\chi}_{\bm{l}}) = \gamma(\bm{\chi}_{\bm{l}}) \quad \forall\, \bm{\chi}_{\bm{l}} \in \partial\D
 \label{eq:proj_sample}
\end{equation}
can be expressed as conditions on each boundary:
\begin{equation}
\begin{aligned}
u(x_1,\ldots,x_{j-1},1,x_{j+1},\ldots,x_d)=\gamma^{j+},\quad \bm{x}\in \partial\D_j, \quad j \in [d]\\
u(x_1,\ldots,x_{j-1},-1,x_{j+1},\ldots,x_d)=\gamma^{j-},\quad \bm{x}\in \partial\D_j, \quad j \in [d].
\label{eq:proj_combination}
\end{aligned}
\end{equation}

\subsubsection{Linear system from the differential equation}
To evaluate the matrix corresponding to the differential operator from \eq{inter_eq}, it is convenient to define coefficients $c^{(\bm{j})}_{\bm{k}}$ and $\|\bm{k}\|_{\infty}\le n$ such that
\begin{equation}
\frac{\partial^{\bm{j}}}{\partial \bm{x}^{\bm{j}}}u(\bm{x})=\sum_{\|\bm{k}\|_{\infty}\le n}c^{(\bm{j})}_{\bm{k}}\phi_{\bm{k}}(\bm{x})
\label{eq:diff_expand}
\end{equation}
for some fixed $\bm{j} \in \N^d$ (as we explain below, such a decomposition exists for the choices of basis functions in \eq{basis_function}).
Using this expression, we obtain the following linear equations for $c^{(\bm{j})}_{\bm{k}}$:
\begin{equation}
\sum_{\|\bm{j}\|_{1}= 2}A_{\bm{j}}\sum_{\|\bm{k}\|_{\infty},\|\bm{l}\|_{\infty}\le n}\phi_{\bm{k}}(\bm{\chi}_{\bm{l}})c^{(\bm{j})}_{\bm{k}}|l_1\rangle\ldots|l_d\rangle=\sum_{\|\bm{k}\|_{\infty},\|\bm{l}\|_{\infty}\le n} \phi_{\bm{k}}(\bm{\chi}_{\bm{l}})\hat{f}_{\bm{k}}|l_1\rangle\ldots|l_d\rangle.
\label{eq:diff_coeff_system}
\end{equation}
To determine the transformation between $c^{(\bm{j})}_{\bm{k}}$ and $c_{\bm{k}}$, we can make use of the differential properties of Fourier and Chebyshev series, namely
\begin{align}
  \frac{\d}{\d x} e^{ik\pi x}=ik\pi e^{ik\pi x}
\end{align}
and
\begin{align}
  2T_k(t)=\frac{T'_{k+1}(t)}{k+1}-\frac{T'_{k-1}(t)}{k-1},
\end{align}
respectively.
We have
\begin{equation}
c^{(\bm{j})}_{\bm{k}}=\sum_{\|\bm{r}\|_{\infty}\le n}[\bm{D}^{(\bm{j})}_{n}]_{\bm{kr}}c_{\bm{r}.},\quad \|\bm{k}\|_{\infty}\le n,
\label{eq:diff_transform}
\end{equation}
where $\bm{D}^{(\bm{j})}_{n}$ can be expressed as the tensor product
\begin{equation}
\bm{D}^{(\bm{j})}_{n} = D^{j_1}_n \otimes D^{j_2}_n \otimes \cdots \otimes D^{j_d}_n,
\label{eq:diff_tensor}
\end{equation}
with $\bm{j}=(j_1,\ldots,j_d)$. The matrix $D_n$ for the Fourier basis functions in \eq{basis_function} can be written as the $(n+1)\times(n+1)$ diagonal matrix with entries
\begin{equation}
[D_n]_{kk}=i(k-\left\lfloor n/2 \right\rfloor)\pi.
\label{eq:Fourier_Dn}
\end{equation}
As detailed in Appendix A of Reference~\cite{CL19}, the matrix $D_n$ for the Chebyshev polynomials in \eq{basis_function} can be expressed as the $(n+1)\times(n+1)$ upper triangular matrix with nonzero entries
\begin{equation}
[D_n]_{kr}=\frac{2r}{\sigma_k},\qquad \text{$k+r$ odd},~ r>k,
\label{eq:Chebyshev_Dn}
\end{equation}
where
\begin{equation}
\sigma_k := \begin{cases}
2 & k=0\\
1 & k\in\range{n}.
\end{cases}
\end{equation}

Substituting \eq{diff_tensor} into \eq{diff_coeff_system}, with $D_n$ defined by \eq{Fourier_Dn} in the periodic case or \eq{Chebyshev_Dn} in the non-periodic case, and performing the multi-dimensional inverse QSFT/QCT (for a reason that will be explained in the next section), we obtain the following linear equations for $c_{\bm{r}}$:
\begin{equation}
\sum_{\|\bm{j}\|_1= 2}A_{\bm{j}}\sum_{\|\bm{k}\|_{\infty},\|\bm{l}\|_{\infty},\|\bm{r}\|_{\infty}\le n}[\bm{D}^{(\bm{j})}_{n}]_{\bm{kr}}c_{\bm{r}}|l_1\rangle\ldots|l_d\rangle=\sum_{\|\bm{k}\|_{\infty},\|\bm{l}\|_{\infty}\le n}\hat{f}_{\bm{k}}|l_1\rangle\ldots|l_d\rangle.
\label{eq:coeff_system}
\end{equation}

Notice that the matrices \eq{Fourier_Dn} and \eq{Chebyshev_Dn} are not full rank. More specifically, there exists at least one zero row in the matrix of \eq{coeff_system} when using either \eq{Fourier_Dn} ($k=\left\lfloor n/2 \right\rfloor$) or \eq{Chebyshev_Dn} ($k=n$). To obtain an invertible linear system, we next introduce the boundary conditions.

\subsubsection{Adding the linear system from the boundary conditions}

When we use the form \eq{u_expand} of $u(\bm{x})$ to write linear equations describing the boundary conditions \eq{proj_combination}, we obtain a non-sparse linear system. Thus, for each $\bm{x}\in \partial\D_j$ in \eq{proj_combination}, we perform the $(d-1)$-dimensional inverse QSFT/QCT on the $d-1$ registers except the $j$th register to obtain the linear equations
\begin{equation}
\begin{aligned}
\sum_{\substack{\|\bm{k}\|_{\infty}\le n \\ k_j=n}}c_{\bm{k}}|k_1\rangle\ldots|k_d\rangle&=\sum_{\substack{\|\bm{k}\|_{\infty}\le n \\ k_j=n}}\hat{\gamma}^{1+}_{\bm{k}}|k_1\rangle\ldots|k_d\rangle,\quad \hat{\gamma}^{j+}_{\bm{k}}\in\partial\D_j, \\
\sum_{\substack{\|\bm{k}\|_{\infty}\le n \\ k_j=n-1}}(-1)^{k_j}c_{\bm{k}}|k_1\rangle\ldots|k_d\rangle&=\sum_{\substack{\|\bm{k}\|_{\infty}\le n \\ k_j=n-1}}\hat{\gamma}^{1-}_{\bm{k}}|k_1\rangle\ldots|k_d\rangle,\quad \hat{\gamma}^{j-}_{\bm{k}}\in\partial\D_j
\label{eq:boundary_linear_system}
\end{aligned}
\end{equation}
for all $j \in [d]$, where the values of $k_j$ indicate that we place these constraints in the last two rows with respect to the $j$th coordinate.
We combine these equations with \eq{coeff_system} to obtain the linear system
\begin{equation}
\sum_{\|\bm{j}\|_1= 2}A_{\bm{j}}\sum_{\|\bm{k}\|_{\infty},\|\bm{r}\|_{\infty}\le n}[\overline{\bm{D}}^{(\bm{j})}_{n}]_{\bm{kr}}c_{\bm{r}}|k_1\rangle\ldots|k_d\rangle=\sum_{\|\bm{k}\|_{\infty}\le n}\sum_{j=1}^d (A_{j,j}\hat{\gamma}^{j+}_{\bm{k}}+A_{j,j}\hat{\gamma}^{j-}_{\bm{k}}+\hat{f}_{\bm{k}})|k_1\rangle\ldots|k_d\rangle,
\label{eq:coeff_linear_system}
\end{equation}
where
\begin{equation}
\overline{\bm{D}}^{(\bm{j})}_{n}=
\begin{cases}
\bm{D}^{(\bm{j})}_{n} + \bm{G}^{(\bm{j})}_{n}, \quad & \|\bm{j}\|_1= 2,~\|\bm{j}\|_{\infty}= 2;\\
\bm{D}^{(\bm{j})}_{n},  & \|\bm{j}\|_1= 2,~\|\bm{j}\|_{\infty}= 1
\end{cases}
\label{eq:boundary_tensor}
\end{equation}
with $\bm{G}^{(\bm{j})}_{n}$ defined below.
In other words, $\overline{\bm{D}}^{(\bm{j})}_{n}=\bm{D}^{(\bm{j})}_{n}+\bm{G}^{(\bm{j})}_{n}$ for each $\bm{j}$ that has exactly one entry equal to $2$ and all other entries $0$, whereas $\overline{\bm{D}}^{(\bm{j})}_{n}=\bm{D}^{(\bm{j})}_{n}$ for each $\bm{j}$ that has exactly two entries equal to $1$ and all other entries $0$. Here $\bm{G}^{(\bm{j})}_{n}$ can be expressed as the tensor product
\begin{equation}
\bm{G}^{(\bm{j})}_{n} = I^{\otimes r-1}\otimes G_n\otimes I^{\otimes d-r}
\label{eq:diff_tensor_G}
\end{equation}
where the $r$th entry of $\bm{j}$ is 2 and all other entries are 0. For the Fourier case in \eq{basis_function} used for periodic boundary conditions, $D_n$ comes from \eq{Fourier_Dn}, and the nonzero entries of $G_n$ are
\begin{equation}
[G_n]_{\left\lfloor n/2 \right\rfloor,k}=1,\quad k \in\rangez{n+1}.
\label{eq:Fourier_Gn}
\end{equation}
Alternatively, for the Chebyshev case in \eq{basis_function} used for non-periodic boundary conditions, $D_n$ comes from \eq{Chebyshev_Dn}, and the nonzero entries of $G_n$ are
\begin{equation}
\begin{aligned}
&[G_n]_{n,k}=1,\quad &k \in\rangez{n+1},\\\
&[G_n]_{n-1,k}=(-1)^k,\quad &k \in\rangez{n+1}.
\label{eq:Chebyshev_Gn}
\end{aligned}
\end{equation}

The system \eq{coeff_linear_system} has the form of \eq{linear_system}. For instance, the matrix in \eq{linear_system} for Poisson's equation \eq{Poisson} is
\begin{align}
L_{\Poisson}&:=\overline{\bm D}_n^{(2,0,\ldots,0)}
+\overline{\bm D}_n^{(0,2,\ldots,0)}
+\cdots+\overline{\bm D}_n^{(0,0,\ldots,2)} \\
&= \bigoplus_{j=1}^d \overline{D}^{(2)}_n=\overline{D}^{(2)}_n\otimes I^{\otimes d-1}+I\otimes \overline{D}^{(2)}_n\otimes I^{\otimes d-2}+\cdots+I^{\otimes d-1}\otimes \overline{D}^{(2)}_n.
\label{eq:matrix_Poisson}
\end{align}
For periodic boundary conditions, using \eq{diff_transform}, \eq{Fourier_Dn}, and \eq{Fourier_Gn}, the second-order differential matrix $\overline{D}^{(2)}_n$ has nonzero entries
\begin{equation}
\begin{aligned}
[\overline{D}^{(2)}_n]_{k,k}&=-((k-\left\lfloor n/2 \right\rfloor)\pi)^2, & k&\in\rangez{n+1}\backslash\{\left\lfloor n/2 \right\rfloor\},\\
[\overline{D}^{(2)}_n]_{\left\lfloor n/2 \right\rfloor,k}&=1, & k&\in\rangez{n+1}.
\label{eq:matrix_Fourier}
\end{aligned}
\end{equation}
For non-periodic boundary conditions, using \eq{diff_transform}, \eq{Chebyshev_Dn}, and \eq{Chebyshev_Gn}, $\overline{D}^{(2)}_n$ has nonzero entries
\begin{equation}
\begin{aligned}
[\overline{D}^{(2)}_n]_{kr}&=\sum_{\substack{l=k+1 \\ \text{$k+l$ odd} \\ \text{$l+r$ odd}}}^{r-1}[D_n]_{kl}[D_n]_{lr}=\frac{2r}{\sigma_k}\sum_{\substack{l=k+1 \\ \text{$k+l$ odd} \\ \text{$l+r$ odd}}}^{r-1}\frac{2l}{\sigma_l}=\frac{r(r^2-k^2)}{\sigma_k}, &&\text{$k+r$ even},\, r>k+1,\\
[\overline{D}^{(2)}_n]_{n,k}&=1, & &k \in\rangez{n+1},\\
[\overline{D}^{(2)}_n]_{n-1,k}&=(-1)^k, & &k \in\rangez{n+1}.
\label{eq:matrix_Chebyshev}
\end{aligned}
\end{equation}

To explicitly illustrate this linear system, we present a simple example in \app{Poisson}. We discuss the invertible linear system \eq{coeff_linear_system} and upper bound its condition number in the following section.


\subsection{Condition number}
\label{sec:condition}

We now analyze the condition number of the linear system. We begin with two lemmas bounding the singular values of the matrices \eq{matrix_Fourier} and \eq{matrix_Chebyshev} that appear in the linear system.

\begin{restatable}{lemma}{svdFourier}\label{lem:svd_Fourier}
Consider the case of periodic boundary conditions. Then for $n\ge4$, the largest and smallest singular values of $\overline{D}^{(2)}_n$ defined in \eq{matrix_Fourier} satisfy
\begin{equation}
\begin{aligned}
&\sigma_{\max}(\overline{D}^{(2)}_n) \le (2n)^{2.5},\\
&\sigma_{\min}(\overline{D}^{(2)}_n) \ge \frac{1}{\sqrt{2}}.
\label{eq:svd_Fourier}
\end{aligned}
\end{equation}
\end{restatable}

\begin{restatable}{lemma}{svdChebyshev}\label{lem:svd_Chebyshev}
Consider the case of non-periodic boundary conditions. Then the largest and smallest singular values of $\overline{D}^{(2)}_n$ defined in \eq{matrix_Chebyshev} satisfy
\begin{equation}
\begin{aligned}
&\sigma_{\max}(\overline{D}^{(2)}_n) \le n^4,\\
&\sigma_{\min}(\overline{D}^{(2)}_n) \ge \frac{1}{16}.
\label{eq:svd_Chebyshev}
\end{aligned}
\end{equation}
\end{restatable}

The proofs of \lem{svd_Fourier} and \lem{svd_Chebyshev} appear in \app{svd}. Using these two lemmas, we first upper bound the condition number of the linear system for Poisson's equation, and then extend the result to general elliptic PDEs.

For the case of the Poisson equation, we use the following simple bounds on the extreme singular values of a Kronecker sum.

\begin{lemma}\label{lem:condition_Poisson}
Let
\begin{equation}
L=\bigoplus_{j=1}^d M_j=M_1\otimes I^{\otimes d-1}+I\otimes M_2\otimes I^{\otimes d-2}+\cdots+I^{\otimes d-1}\otimes M_d,
\label{eq:tensorsum}
\end{equation}
where $\{M_j\}_{j=1}^d$ are square matrices. If the largest and smallest singular values of $M_j$ satisfy
\begin{equation}
\begin{aligned}
&\sigma_{\max}(M_j) \le s_j^{\max},\\
&\sigma_{\min}(M_j) \ge s_j^{\min},
\label{eq:svd_tensorsum}
\end{aligned}
\end{equation}
respectively, then the condition number of $L$ satisfies
\begin{equation}
\kappa_L\le \frac{\sum_{j=1}^ds_j^{\max}}{\sum_{j=1}^ds_j^{\min}}.
\label{eq:con_tensorsum}
\end{equation}
\end{lemma}

\begin{proof}
We bound the singular values of the matrix exponential $\exp(M_j)$ by
\begin{equation}
\begin{aligned}
&\sigma_{\max}(\exp(M_j)) \le e^{s_j^{\max}},\\
&\sigma_{\min}(\exp(M_j)) \ge e^{s_j^{\min}}
\label{eq:svd_exp}
\end{aligned}
\end{equation}
using \eq{svd_tensorsum}. The singular values of the Kronecker product $\bigotimes_{j=1}^d\exp(M_j)$ are
\begin{equation}
\sigma_{k_1,\ldots,k_d}\biggl(\bigotimes_{j=1}^d\exp(M_j)\biggr)=\prod_{j=1}^d\sigma_{k_j}(\exp(M_j))
\end{equation}
where $\sigma_{k_j}(\exp(M_j))$ are the singular values of the matrix $\exp(M_j)$ for each $j \in \range{d}$, where $k_j$ runs from $1$ to the dimension of $M_j$. Using the property of the Kronecker sum that
\begin{equation}
\exp(L)=\exp\biggl(\bigoplus_{j=1}^d M_j\biggr)=\bigotimes_{j=1}^d\exp(M_j),
\label{eq:Kronecker}
\end{equation}
we bound the singular values of the matrix exponential of \eq{matrix_Poisson} by
\begin{equation}
\begin{aligned}
&\sigma_{\max}(\exp(L)) \le e^{\sum_{j=1}^ds_j^{\max}},\\
&\sigma_{\min}(\exp(L)) \ge e^{\sum_{j=1}^ds_j^{\min}}.
\label{eq:svd_prod}
\end{aligned}
\end{equation}
Finally, we bound the singular values of the matrix logarithm of \eq{svd_prod} by
\begin{equation}
\begin{aligned}
&\sigma_{\max}(L) \le \sum_{j=1}^ds_j^{\max},\\
&\sigma_{\min}(L) \ge \sum_{j=1}^ds_j^{\min}.
\label{eq:svd_Poisson}
\end{aligned}
\end{equation}
Thus the condition number of $L$ satisfies
\begin{equation}
\kappa_L \le \frac{\sum_{j=1}^ds_j^{\max}}{\sum_{j=1}^ds_j^{\min}}
\end{equation}
as claimed.
\end{proof}

This lemma easily implies a bound on the condition number of the linear system for Poisson's equation:

\begin{corollary}\label{cor:condition_Poisson}
Consider an instance of the quantum PDE problem as defined in \prb{pde} for Poisson's equation \eq{Poisson} with Dirichlet boundary conditions \eq{dbc}. Then for $n\ge4$, the condition number of $L_{\Poisson}$ in the linear system \eq{linear_system} satisfies
\begin{equation}
\kappa_{L_{\Poisson}}\le (2n)^4.
\end{equation}
\end{corollary}

\begin{proof}
The matrix in \eq{linear_system} for Poisson's equation \eq{Poisson} is $L_{\Poisson}$ defined in \eq{matrix_Poisson}.
For both the periodic and the non-periodic case, we have
\begin{equation}
\begin{aligned}
&\sigma_{\max}(\overline{D}^{(2)}_n) \le n^4,\\
&\sigma_{\min}(\overline{D}^{(2)}_n) \ge \frac{1}{16}
\label{eq:svd_Fourier_Chebyshev}
\end{aligned}
\end{equation}
by \lem{svd_Fourier} and \lem{svd_Chebyshev}. Let $M_j=\overline{D}^{(2)}_n$ for $j\in\range{d}$ in \eq{tensorsum}, and apply \lem{condition_Poisson} with $s_j^{\max}=n^4$ and $s_j^{\min}={1}/{16}$ in \eq{con_tensorsum}. Then the condition number of $L_{\Poisson}$ is bounded by
\begin{equation}
\kappa_{L_{\Poisson}} \le \frac{\sigma_{\max}(\overline{D}^{(2)}_n)}{\sigma_{\min}(\overline{D}^{(2)}_n)} \le (2n)^4
\end{equation}
as claimed.
\end{proof}

We now consider the condition number of the linear system for general elliptic PDEs.

\begin{lemma}\label{lem:condition}
Consider an instance of the quantum PDE problem as defined in \prb{pde} with Dirichlet boundary conditions \eq{dbc}.
Then for $n\ge4$, the condition number of $L$ in the linear system \eq{linear_system} satisfies
\begin{equation}
\kappa_L\le \frac{\|A\|_\Sigma}{C\|A\|_{\ast}}(2n)^4,
\end{equation}
where $\|A\|_{\Sigma}:=\sum_{\|\bm{j}\|_1\le 2} |A_{\bm{j}}|=\sum_{j_1,j_2=1}^d |A_{j_1,j_2}|$, $\|A\|_{\ast}:=\sum_{j=1}^d |A_{j,j}|$, and $C>0$ is defined in \eq{DDM}.
\end{lemma}

Recall that $C$ quantifies the extent to which the global strict diagonal dominance condition holds.

\begin{proof}
According to \eq{coeff_linear_system}, the matrix in \eq{linear_system} is
\begin{equation}
L=\sum_{\|\bm{j}\|_1= 2}A_{\bm{j}}\overline{\bm{D}}^{(\bm{j})}_{\bm{n}}.
\label{eq:matrix_L}
\end{equation}
We upper bound the spectral norm of the matrix $L$ by
\begin{equation}
\|L\| \le \sum_{\|\bm{j}\|_1= 2} |A_{\bm{j}}| \|\overline{\bm{D}}^{(\bm{j})}_{\bm{n}}\|.
\end{equation}
For the matrix $\overline{\bm{D}}^{(\bm{j})}_{\bm{n}}$ defined by \eq{boundary_tensor}, \lem{svd_Fourier} (in the periodic case) and \lem{svd_Chebyshev} (in the non-periodic case) give the inequality
\begin{equation}
\|\overline{\bm{D}}^{(\bm{j})}_{\bm{n}}\|\le n^4,
\label{eq:bound_Dn}
\end{equation}
so we have
\begin{equation}
\|L\| \le \sum_{\|\bm{j}\|_1= 2} |A_{\bm{j}}| n^4 = \|A\|_\Sigma \, n^4.
\label{eq:matrix_norm}
\end{equation}

Next we lower bound $\|L\xi\|$ for any $\|\xi\|=1$.

It is non-trivial to directly compute the singular values of a sum of non-normal matrices. Instead, we write $L$ as a sum of terms $L_1$ and $L_2$, where $L_1$ is a tensor sum similar to \eq{matrix_Poisson} that can be bounded by \lem{condition_Poisson}, and $L_2$ is a sum of tensor products that are easily bounded. Specifically, we have
\begin{equation}
\begin{aligned}
L_1&=A_{1,1}\overline{D}^{(2)}_n\otimes I^{\otimes d-1}+\cdots+A_{d,d}I^{\otimes d-1}\otimes \overline{D}^{(2)}_n\\
L_2&=L-L_1.
\end{aligned}
\end{equation}
The ellipticity condition \eq{elliptic}, $\forall\bm{\xi}\,\sum_{\|\bm{j}\|_1= h}A_{\bm{j}}(\bm{x})\bm{\xi}^{\bm{j}}\ne0$, can only hold if the $A_{j,j}$ for $j \in \range{d}$ are either all positive or all negative; we consider $A_{j,j}>0$ without loss of generality, so
\begin{equation}
\|A\|_{\ast}=\sum_{j=1}^d |A_{j,j}|=\sum_{j=1}^dA_{j,j}.
\end{equation}
Also, the global strict diagonal dominance condition \eq{DDM} simplifies to
\begin{equation}
C=1-\sum_{j_1=1}^d\frac{1}{A_{j_1,j_1}}\sum_{j_2\in\range{d}\backslash\{j_1\}} |A_{j_1,j_2}|>0,
\label{eq:discrete_uniform}
\end{equation}
where $0 < C \le 1$.

We now upper bound $\|L_2 L_1^{-1}\|$ by bounding $\|\bm{D}^{(\bm{j})}_{\bm{n}} L_1^{-1}\|$ for each $\bm{j}=(j_1, \ldots, j_d)$ that has exactly two entries equal to $1$ and all other entries $0$. Specifically, consider $j_{r_1}=j_{r_2}=1$ for $r_1, r_2 \in \range{d}$, $r_1\ne r_2$, and $j_{r}=0$ for $r \in \range{d}\backslash\{r_1, r_2\}$. We denote
\begin{equation}
L^{(\bm{j})}:=I^{\otimes r_1-1}\otimes D^{2}_n\otimes I^{\otimes d-r_1}+I^{\otimes r_2-1}\otimes D^{2}_n\otimes I^{\otimes d-r_2}.
\end{equation}
We first upper bound $\|\bm{D}^{(\bm{j})}_{\bm{n}}\|$ by $\frac{1}{2}\|L^{(\bm{j})}\|$. Notice the matrices $\bm{D}^{(\bm{j})}_{\bm{n}}$ and $L^{(\bm{j})}$ share the same singular vectors. For $k \in \rangez{n+1}$, we let $v_k$ and $\lambda_k$ denote the right singular vectors and corresponding singular values of $D_n$, respectively. Then the right singular vectors of $\bm{D}^{(\bm{j})}_{\bm{n}}$ and $L^{(\bm{j})}$ are $\bm{v}_{\bm{k}}:=\bigotimes_{j=1}^d v_{k_j}$, where $\bm{k}=(k_1,\ldots,k_d)$ with $k_j \in \rangez{n+1}$ for $j \in \range{d}$. For any vector $v=\sum_{\|\bm{k}\|_\infty \le n}\alpha_{\bm{k}}\bm{v}_{\bm{k}}$, we have
\begin{equation}
\|\bm{D}^{(\bm{j})}_{\bm{n}}v\|^2 = \sum_{\|\bm{k}\|_\infty \le n}|\alpha_{\bm{k}}|^2\|\bm{D}^{(\bm{j})}_{\bm{n}}\bm{v}_{\bm{k}}\|^2 = \sum_{\|\bm{k}\|_\infty \le n}|\alpha_{\bm{k}}|^2(\lambda_{k_{j_{r_1}}}\lambda_{k_{j_{r_2}}})^2,
\label{eq:svd_Dj}
\end{equation}
\begin{equation}
\|L^{(\bm{j})}v\|^2 = \sum_{\|\bm{k}\|_\infty \le n}|\alpha_{\bm{k}}|^2\|L^{(\bm{j})}\bm{v}_{\bm{k}}\|^2 = \sum_{\|\bm{k}\|_\infty \le n}|\alpha_{\bm{k}}|^2(\lambda_{k_{j_{r_1}}}^2+\lambda_{k_{j_{r_2}}}^2)^2,
\label{eq:svd_Lj}
\end{equation}
which implies $\|\bm{D}^{(\bm{j})}_{\bm{n}}v\| \le \frac{1}{2}\|L^{(\bm{j})}v\|$ by the inequality of arithmetic and geometric means (also known as the AM-GM inequality). Since this holds for any vector $v$, we have
\begin{equation}
\|\bm{D}^{(\bm{j})}_{\bm{n}}L_1^{-1}\| \le \frac{1}{2}\|L^{(\bm{j})}L_1^{-1}\|.
\end{equation}

Next we upper bound $\|D^{2}_n\|$ by $\|\overline{D}^{(2)}_n\|$. For any vector $u=[u_0,\ldots,u_n]^T$, define two vectors $w=[w_0,\ldots,w_n]^T$ and $\overline{w}=[\overline{w}_0,\ldots,\overline{w}_n]^T$ such that
\begin{equation}
D^{2}_n[u_0,\ldots,u_n]^T = [w_0,\ldots,w_n]^T
\end{equation}
and
\begin{equation}
\overline{D}^{2}_n[u_0,\ldots,u_n]^T = [\overline{w}_0,\ldots,\overline{w}_n]^T.
\end{equation}
Notice that $w_{\left\lfloor n/2 \right\rfloor}=0$ and $w_k=\overline{w}_k$ for $k \in \rangez{n+1}\backslash\{\left\lfloor n/2 \right\rfloor\}$ for periodic conditions, and $w_{n-1}=w_n=0$ and $w_k=\overline{w}_k$ for $k \in \rangez{n+1}\backslash\{n-1,n\}$ for non-periodic conditions. Thus, for any vector $v$,
\begin{equation}
\|D^{2}_nv\|^2 =\|w\|^2 = \sum_{k=0}^nw_k^2 \le \sum_{k=0}^n\overline{w}_k^2 = \|\overline{w}\|^2 = \|\overline{D}^{(2)}_nv\|^2.
\end{equation}
Therefore,
\begin{equation}
\|L^{(\bm{j})}L_1^{-1}\| \le \sum_{s=1}^2\|I^{\otimes r_s-1}\otimes D^{2}_n\otimes I^{\otimes d-r_s}L_1^{-1}\| \le \sum_{s=1}^2\|I^{\otimes r_s-1}\otimes \overline{D}^{(2)}_n\otimes I^{\otimes d-r_s}L_1^{-1}\|.
\end{equation}
We also have
\begin{equation}
\|\bm{D}^{(\bm{j})}_{\bm{n}} L_1^{-1}\| \le \frac{1}{2}\sum_{s=1}^2\|I^{\otimes r_s-1}\otimes \overline{D}^{(2)}_n\otimes I^{\otimes d-r_s}L_1^{-1}\|.
\end{equation}
We can rewrite $I^{\otimes r_s-1}\otimes \overline{D}^{(2)}_n\otimes I^{\otimes d-r_s}L_1^{-1}$ in the form
\begin{equation}
I^{\otimes r_s-1}\otimes \overline{D}^{(2)}_n\otimes I^{\otimes d-r_s}\left(\sum_{h=1}^dA_{h,h}I^{\otimes r_h-1}\otimes \overline{D}^{(2)}_n\otimes I^{\otimes d-r_h}\right)^{-1}.
\end{equation}
The matrices $I^{\otimes r_h-1}\otimes \overline{D}^{(2)}_n\otimes I^{\otimes d-r_h}$ share the same singular values and singular vectors, so
\begin{align}
\|I^{\otimes r_s-1}\otimes \overline{D}^{(2)}_n\otimes I^{\otimes d-r_s}L_1^{-1}\| =  \max_{\overline{\lambda}_{k_r}}\frac{\overline{\lambda}_{k_{r_s}}}{\sum_{h=1}^d A_{h,h}\overline{\lambda}_{k_h}} < \frac{1}{A_{{r_s},{r_s}}},
\end{align}
where $\overline{\lambda}_{k_h}$ are singular values of $I^{\otimes r_h-1}\otimes \overline{D}^{(2)}_n\otimes I^{\otimes d-r_h}$ for $k_h \in \rangez{n}$, $h\in \range{d}$. This implies
\begin{equation}
\|\bm{D}^{(\bm{j})}_{\bm{n}} L_1^{-1}\| \le \frac{1}{2}(\frac{1}{A_{{r_1},{r_1}}}+\frac{1}{A_{{r_2},{r_2}}}).
\end{equation}
Using \eq{discrete_uniform}, considering each instance of $\bm{D}^{(\bm{j})}_{\bm{n}}$ in $L_2$, we have
\begin{equation}
\|L_2L_1^{-1}\| \le \sum_{j_1\ne j_2}|A_{j_1,j_2}| \|\bm{D}^{(\bm{j})}_{\bm{n}} L_1^{-1}\| \le \sum_{j_1=1}^d\frac{1}{A_{j_1,j_1}}\sum_{j_2\in\range{d}\backslash\{j_1\}} |A_{j_1,j_2}| \le 1-C.
\end{equation}
Since $L$ and $L_1$ are invertible, $\|L_1^{-1}L_2\|\le 1-C<1$, and by \lem{condition_Poisson} applied to $\|L_1^{-1}\|$, we have
\begin{equation}
\|L^{-1}\|=\|(L_1+L_2)^{-1}\| \le \|(I+L_2L_1^{-1})^{-1}\| \|L_1^{-1}\| \le \frac{\|L_1^{-1}\|}{1-\|L_2L_1^{-1}\|}\le\frac{1/\frac{1}{16}\|A\|_{\ast}}{C}=\frac{16}{C\|A\|_{\ast}}.
\end{equation}
Thus we have
\begin{equation}
\kappa_L = \|L\|\|L^{-1}\| \le \frac{\|A\|_\Sigma}{C\|A\|_{\ast}}(2n)^4
\end{equation}
as claimed.
\end{proof}


\subsection{State preparation}
\label{sec:state}

We now describe a state preparation procedure for the vector $\vec f + \vec g$ in the linear system \eq{linear_system}.

\begin{lemma}\label{lem:preparation}
Let $O_f$ be a unitary oracle that maps $|0\rangle|0\rangle$ to a state proportional to $|0\rangle|f\rangle$, and $|\phi\rangle|0\rangle$ to $|\phi\rangle|0\rangle$ for any $|\phi\rangle$ orthogonal to $|0\rangle$; let $O_x$ be a unitary oracle that maps $|0\rangle|0\rangle$ to $|0\rangle|0\rangle$, $|j\rangle|0\rangle$ to a state proportional to $|j\rangle|\gamma^{j+}\rangle$ for $j \in \range{d}$, and $|j+d\rangle|0\rangle$ to a state proportional to $|j+d\rangle|\gamma^{j-}\rangle$ for $j \in \range{d}$. Suppose $\norm{|f\rangle},\norm{|\gamma^{j+}\rangle},\norm{|\gamma^{j-}\rangle}$ and $A_{j,j}$ for $j \in \range{d}$ are known.
Define the parameter
\begin{equation}
  q: = \sqrt{\frac{\sum_{\|\bm{k}\|_{\infty}\le n} \sum_{j=1}^d [ \hat{f}^2_{\bm{k}}+(A_{j,j}\hat{\gamma}^{j+}_{\bm{k}})^2+(A_{j,j}\hat{\gamma}^{j-}_{\bm{k}})^2]}{\sum_{\|\bm{k}\|_{\infty}\le n} \sum_{j=1}^d |\hat{f}_{\bm{k}}+A_{j,j}\hat{\gamma}^{j+}_{\bm{k}}+A_{j,j}\hat{\gamma}^{j-}_{\bm{k}}|^2}}.
\end{equation}
Then the normalized quantum state
\begin{equation}
  |B\rangle \propto
\sum_{\|\bm{k}\|_{\infty}\le n} \sum_{j=1}^d (\hat{f}_{\bm{k}}+A_{j,j}\hat{\gamma}^{j+}_{\bm{k}}+A_{j,j}\hat{\gamma}^{j-}_{\bm{k}})|k_1\rangle\ldots|k_d\rangle,
\end{equation}
with coefficients defined as in \eq{f_coeff} and \eq{g_coeff},
can be prepared with gate and query complexity $O(qd^2\log n\log\log n)$.
\end{lemma}

\begin{proof}
Starting from the initial state $|0\rangle|0\rangle$, we first perform a unitary transformation $U$ satisfying
\begin{equation}
\begin{aligned}
  U|0\rangle=&\frac{\||f\rangle\|}{\sqrt{\||f\rangle\|^2
  +\sum_{j=1}^d\bigl(A_{j,j}^2\||\gamma^{j+}\rangle\|^2+A_{j,j}^2\||\gamma^{j-}\rangle\|^2\bigr)}}|0\rangle\\
  &+\sum_{j=1}^d\frac{A_{j,j}\||\gamma^{j+}\rangle\|}{\sqrt{\||f\rangle\|^2
  +\sum_{j=1}^d\bigl(A_{j,j}^2\||\gamma^{j+}\rangle\|^2+A_{j,j}^2\||\gamma^{j-}\rangle\|^2\bigr)}}|j\rangle\\
  &+\sum_{j=1}^d\frac{A_{j,j}\||\gamma^{j-}\rangle\|}{\sqrt{\||f\rangle\|^2
  +\sum_{j=1}^d\bigl(A_{j,j}^2\||\gamma^{j+}\rangle\|^2+A_{j,j}^2\||\gamma^{j-}\rangle\|^2\bigr)}}|j+d\rangle
\end{aligned}
\end{equation}
on the first register to obtain
\begin{equation}
  \frac{\||f\rangle\||0\rangle
  +A_{1,1}\||\gamma^{1+}\rangle\||1\rangle+\cdots+A_{d,d}\||\gamma^{d-}\rangle\| |2d\rangle}
  {\sqrt{\||f\rangle\|^2
  +\sum_{j=1}^d\bigl(A_{j,j}^2\||\gamma^{j+}\rangle\|^2+A_{j,j}^2\||\gamma^{j-}\rangle\|^2\bigr)}}|0\rangle.
\end{equation}
This can be done in time $O(2d+1)$ by standard techniques \cite{SBM06}. Then we apply $O_x$ and $O_f$ to obtain
\begin{equation}
  \begin{aligned}
  &|0\rangle|f\rangle+A_{1,1}|1\rangle|\gamma^{1+}\rangle+\cdots+A_{d,d}|2d\rangle|\gamma^{d-}\rangle, \\
&\quad\propto
\sum_{\|\bm{k}\|_{\infty},\|\bm{l}\|_{\infty}\le n} \phi_{\bm{k}}(\bm{\chi}_{\bm{l}})(\hat{f}_{\bm{k}}|0\rangle+A_{1,1}\hat{\gamma}^{1+}_{\bm{k}}|1\rangle+\cdots+A_{d,d}\hat{\gamma}^{d-}_{\bm{k}}|2d\rangle)|l_1\rangle\ldots|l_d\rangle,
\end{aligned}
\end{equation}
according to \eq{f_inter} and \eq{g_inter}. We then perform the $d$-dimensional inverse QSFT (for periodic boundary conditions) or inverse QCT (for non-periodic boundary conditions) on the last $d$ registers, obtaining
\begin{equation}
\sum_{\|\bm{k}\|_{\infty}\le n} (\hat{f}_{\bm{k}}|0\rangle+A_{1,1}\hat{\gamma}^{1+}_{\bm{k}}|1\rangle+\cdots+A_{d,d}\hat{\gamma}^{d-}_{\bm{k}}|2d\rangle)|k_1\rangle\ldots|k_d\rangle.
\end{equation}
Finally, observe that if we measure the first register in a basis containing the uniform superposition $|0\rangle+|1\rangle+\cdots+|2d\rangle$ (say, the Fourier basis) and obtain the outcome corresponding to the uniform superposition, we produce the state
\begin{equation}
\sum_{\|\bm{k}\|_{\infty}\le n} \sum_{j=1}^d (\hat{f}_{\bm{k}}+A_{j,j}\hat{\gamma}^{j+}_{\bm{k}}+A_{j,j}\hat{\gamma}^{j-}_{\bm{k}})|k_1\rangle\ldots|k_d\rangle.
\end{equation}
Since this outcome occurs with probability $1/q^2$, we can prepare this state with probability close to $1$ using $O(q)$ steps of amplitude amplification. According to \lem{qsft} and \lem{qct}, the $d$-dimensional (inverse) QSFT or QCT can be performed with gate complexity $O(d\log n\log\log n)$. Thus the total gate and query complexity is $O(qd^2\log n\log\log n)$.
\end{proof}

Alternatively, if it is possible to directly prepare the quantum state $|B\rangle$, then we may be able to avoid the factor of $q$ in the complexity of the overall algorithm.


\subsection{Main result}
\label{sec:main}

Having analyzed the condition number and state preparation procedure for our approach, we are now ready to establish the main result. \thm{main} as follows.

\theoremSpectral*

\begin{proof}
We analyze the complexity of the algorithm presented in \sec{solver}.

First we choose
\begin{equation}
n:=\biggl\lfloor\frac{\log(\Omega)}{\log(\log(\Omega))}\biggr\rfloor,
\label{eq:truncated_number}
\end{equation}
where
\begin{equation}
\Omega=\frac{g'(1+\epsilon)}{g\epsilon}.
\label{eq:Omega}
\end{equation}
By Eq.~(1.8.28) of Reference~\cite{Tang06}, this choice guarantees
\begin{equation}
\|\hat{u}(\bm{x})-u(\bm{x})\|\le
\max_{\bm{x}}\|\hat{u}^{(n+1)}(\bm{x})\|\frac{e^{n}}{(2n)^n}\le\frac{g'}{\Omega}=\frac{g\epsilon}{1+\epsilon}=:\delta.
\label{eq:n_inequality}
\end{equation}
Now $\|\hat{u}(\bm{x})-u(\bm{x})\|\le\delta$ implies
\begin{equation}
\biggl\|\frac{\hat{u}(\bm{x})}{\|\hat{u}(\bm{x})\|}-\frac{u(\bm{x})}{\|u(\bm{x})\|}\biggr\|
\le\frac{\delta}{\min\{\|\hat{u}(\bm{x})\|,\|u(\bm{x})\|\}}\le\frac{\delta}{g-\delta}=\epsilon,
\label{eq:normalized_inequality}
\end{equation}
so we can choose $n$ to ensure that the normalized output state is $\epsilon$-close to $\hat{u}(\bm{x})/\| \hat{u}(\bm{x})\|$.

As described in \sec{solver}, the algorithm uses the high-precision QLSA from Reference~\cite{CKS15} and the multi-dimensional QSFT/QCT (and its inverse). According to \lem{qsft} and \lem{qct}, the $d$-dimensional (inverse) QSFT or QCT can be performed with gate complexity $O(d\log n\log\log n)$. According to \lem{preparation}, the query and gate complexity for state preparation is $O(qd^2\log n\log\log n)$.

For the linear system $L \vec x=\vec f + \vec g$ in \eq{linear_system}, the matrix $L$ is an $(n+1)^d \times (n+1)^d$ matrix with $(n+1)$ or $(n+1)d$ nonzero entries in any row or column for periodic or non-periodic conditions, respectively. According to \lem{condition}, the condition number of $L$ is upper bounded by $\frac{\|A\|_\Sigma}{C\|A\|_{\ast}}(2n)^4$. Consequently, by Theorem 5 of Reference~\cite{CKS15}, the QLSA produces a state proportional to $\vec x$ with $O(\frac{d\|A\|_\Sigma}{C\|A\|_{\ast}}(2n)^5)$ queries to the oracles, and its gate complexity is larger by a factor of $\poly(\log(d\|A\|_\Sigma\,n))$. Using the value of $n$ specified in \eq{truncated_number}, the overall query complexity of our algorithm is
\begin{equation}
\biggl(\frac{d\|A\|_\Sigma}{C\|A\|_{\ast}}+qd^2\biggr)\poly(\log(g'/g\epsilon)),
\end{equation}
and the gate complexity is
\begin{equation}
\biggl(\frac{d\|A\|_\Sigma}{C\|A\|_{\ast}}\poly(\log(d\|A\|_\Sigma/\epsilon))+qd^2\biggr)\poly(\log(g'/g\epsilon))
\end{equation}
which is larger by a factor of $\poly(\log(d\|A\|_\Sigma/\epsilon))$, as claimed.
\end{proof}

Note that we can establish a more efficient algorithm in the special case of the Poisson equation with homogenous boundary conditions. In this case, $\|A\|_\Sigma = \|A\|_{\ast} = d$ and $C=1$. Under homogenous boundary conditions, the complexity of state preparation can be reduced to $d\poly(\log(g'/g\epsilon))$, since we can remove $2d$ applications of the QSFT or QCT for preparing a state depending on the boundary conditions, and since $\gamma=0$ there is no need to postselect on the uniform superposition to incorporate the boundary conditions. In summary, the query complexity of the Poisson equation with homogenous boundary conditions is
\begin{equation}
d\poly(\log(g'/g\epsilon));
\end{equation}
again the gate complexity is larger by a factor of
$\poly(\log(d\|A\|_\Sigma/\epsilon))$.


\section{Discussion and open problems}
\label{sec:discussion}

We have presented high-precision quantum algorithms for $d$-dimensional PDEs using the FDM and spectral methods. These algorithms use high-precision QLSAs to solve Poisson's equation and other second-order elliptic equations. Whereas previous algorithms scaled as $\poly (d, 1 / \epsilon)$, our algorithms scale as $\poly (d, \log (1 / \epsilon ))$.

This work raises several natural open problems. First, for the quantum adaptive FDM, we only deal with Poisson's equation with homogeneous boundary conditions. Can we apply the adaptive FDM to other linear equations or to inhomogeneous boundary conditions?
The quantum spectral algorithm applies to second-order elliptic PDEs with Dirichlet boundary conditions.
Can we generalize it to other linear PDEs with Neumann or mixed boundary conditions? Also, can we develop algorithms for space- and time-dependent PDEs?
These cases are more challenging since the quantum Fourier transform cannot be directly applied to ensure sparsity.
Finally, can we improve the dependence on $d$?

Second, the complexity scales logarithmically with high-order derivatives (of the inhomogeneity or solution) for both the adaptive FDM and the spectral method. In particular, \thm{main_fdm} shows that the complexity of the quantum adaptive FDM scales logarithmically with $\bigl| \frac{\d^{2k+1} u}{\d x^{2k+1}} \bigr|$, and \thm{main} shows that the complexity of the quantum spectral method is $\poly(\log g')$, where $g'$ upper bounds $\|\hat{u}^{(n+1)}(\bm{x})\|$ (see \eq{main_thm_g}).
Such a logarithmic dependence on high-order derivatives of the solution is typical for classical algorithms, including the classical adaptive FDM (see for example Theorem 7 of Reference~\cite{kivlichan2017bounding}) and spectral methods (see for example Eq.~(1.8.28) of Reference~\cite{Tang06}), both of which have the same logarithmic dependence on $\bigl| \frac{\d^{2k+1} u}{\d x^{2k+1}} \bigr|$ and $g'$.
This logarithmic dependence means that the algorithm is efficient even when faced with a highly oscillatory solution with an exponentially large derivative.

However, the query complexity of time-dependent Hamiltonian simulation only depends on the first-order derivatives of the Hamiltonian \cite{PQS11,BCC13}. Can we develop quantum algorithms for PDEs with query complexity independent of high-order derivatives, and henceforth develop an unexpected advantage of quantum algorithms for PDEs?

Third, can we develop quantum algorithms for other types of PDEs, such as stochastic or nonlinear PDEs?

Fourth, can we use quantum algorithms for PDEs as a subroutine of other quantum algorithms? For example, some PDE algorithms have state preparation steps that require inverting finite difference matrices (such as Reference~\cite{CJO17} using certain oracles for the initial conditions); are there other scenarios in which state preparation can be done using the solution of another system of PDEs? Can quantum algorithms for PDEs be applied to other algorithmic tasks, such as optimization?

Finally, how should these algorithms be applied? While PDEs have broad applications, much more work remains to understand the extent to which quantum algorithms can be of practical value. Answering this question will require careful consideration of various technical aspects of the algorithms. In particular: What measurements give useful information about the solutions, and how can those measurements be efficiently implemented? How should the oracles encoding the equations and boundary conditions be implemented in practice? And with these aspects taken into account, what are the resource requirements for quantum computers to solve classically intractable problems related to PDEs?


\section*{Acknowledgments}

JPL thanks Jacob Bedrossian, Dominic Berry, Stephen Jordan, Ashley Montanaro, and Konstantina Trivisa for valuable discussions.
We also thank Tongyang Li for pointing out an issue with the analysis of the adaptive FDM algorithm in a previous version of the paper,
and we thank anonymous referees for their comments on an earlier draft.

The authors acknowledge support from National Science Foundation grant CCF-1813814 and from the U.S.\ Department of Energy, Office of Science, Office of Advanced Scientific Computing Research, Quantum Algorithms Teams and Accelerated Research in Quantum Computing programs.


\providecommand{\bysame}{\leavevmode\hbox to3em{\hrulefill}\thinspace}

\newpage
\appendix


\section{An example for solving Poisson's equation}
\label{app:Poisson}

In this appendix, we present an example of solving Poisson's equation in two dimensions using our algorithm. The Poisson equation is
\begin{align}
\Delta u(x_1,x_2)&=\left(\frac{\partial^2}{\partial x_1^2}+\frac{\partial^2}{\partial x_2^2}\right)u(x_1,x_2)=f(x_1,x_2), & (x_1,x_2) &\in \Omega=[-1, 1]^2.
\label{eq:Poisson_eq}
\end{align}

We consider two kinds of boundary value problems, as follows.
\begin{itemize}
  \item \emph{Periodic boundary conditions:}
  \begin{equation}
  \begin{aligned}
  u(x_1,x_2) &= u(x_1+2v,x_2) = u(x_1,x_2+2v) , & (x_1,x_2)& \in \Omega=[-1, 1]^2, v \in \Z\\
  u(0,0) &= \gamma.
  \label{eq:period_cond}
  \end{aligned}
  \end{equation}

  \item \emph{Non-periodic boundary conditions:}
  \begin{equation}
  \begin{aligned}
  u(x_1,1) &= \gamma_N(x_1),  & u(1,x_2) &= \gamma_E(x_2),\\
  u(x_1,-1) &= \gamma_S(x_1), & u(-1,x_2) &= \gamma_W(x_2).
  \label{eq:non_period_cond}
  \end{aligned}
  \end{equation}
\end{itemize}

We first present the quantum Fourier spectral method to solve \eq{Poisson_eq} with the periodic conditions \eq{period_cond}. In particular, we choose $n=2$ in the specification of the linear system. The truncated Fourier series can be written as
\begin{equation}
u(x_1,x_2) = \sum_{k_1=0}^2\sum_{k_2=0}^2 c_{k_1,k_2}e^{i(k_1-1)\pi x_1}e^{i(k_2-1)\pi x_2} .
\label{eq:period_expand}
\end{equation}
We are given an oracle for preparing the state
\begin{equation}
\sum_{l_1=0}^2\sum_{l_2=0}^2f\left(\frac{2 l_1}{3}-1,\frac{2 l_2}{3}-1\right)|l_1\rangle|l_2\rangle
\label{eq:period_inter}
\end{equation}
that interpolates the uniform grid nodes \eq{interpolation_nodes}.
We first perform a multi-dimensional inverse QSFT on \eq{period_inter} to obtain
\begin{equation}
\sum_{k_1=0}^2\sum_{k_2=0}^2f_{k_1,k_2}|k_1\rangle|k_2\rangle,
\label{eq:Poisson_coeff}
\end{equation}
where $b_{k_1,k_2}$ are defined by \eq{f_coeff}. Then we apply the quantum linear system algorithm of Reference~\cite{CKS15} to solve the linear system
\begin{equation}
L_p|X\rangle=|B\rangle,
\label{eq:linear_system_Poisson}
\end{equation}
where the solution is
\begin{equation}
|X\rangle=  \sum_{k_1=0}^2\sum_{k_2=0}^2 c_{k_1,k_2}|k_1\rangle|k_2\rangle.
\label{eq:solution_Poisson}
\end{equation}
According to \eq{coeff_linear_system},the discretized linear system from \eq{Poisson_eq} is
\begin{equation}
D^{2}_n\otimes I+I\otimes D^{2}_n,
\label{eq:Poisson_matrix}
\end{equation}
where the Fourier difference matrix $D_n$ is defined by \eq{Fourier_Dn} with $n=2$, namely
\begin{equation}
D_2=
  \begin{pmatrix}
    -i\pi & 0 & 0 \\
    0 & 0 & 0 \\
    0 & 0 & i\pi \\
  \end{pmatrix},
\end{equation}
so that
\begin{equation}
D_2^2=
  \begin{pmatrix}
    -\pi^2 & 0 & 0 \\
    0 & 0 & 0 \\
    0 & 0 & -\pi^2 \\
  \end{pmatrix}.
\end{equation}
Therefore, the matrix \eq{Poisson_matrix} is
\begin{equation}
D^{2}_2\otimes I+I\otimes D^{2}_2=
  \begin{pmatrix}
    -2\pi^2 &  &  &  &  &  &  &  &  \\
     & -\pi^2 &  &  &  &  &  &  &  \\
     &  & -2\pi^2 &  &  &  &  &  &  \\
     &  &  & -\pi^2 &  &  &  &  &  \\
     &  &  &  & 0 &  &  &  &  \\
     &  &  &  &  & -\pi^2 &  &  &  \\
     &  &  &  &  &  & -2\pi^2 &  &  \\
     &  &  &  &  &  &  & -\pi^2 &  \\
     &  &  &  &  &  &  &  & -2\pi^2 \\
  \end{pmatrix}.
\end{equation}
The rank of this matrix is $(n+1)^d-1$ with $d=2$, $n=2$. We use the boundary condition to complete the linear system:
\begin{equation}
  \begin{pmatrix}
    -2\pi^2 &  &  &  &  &  &  &  &  \\
     & -\pi^2 &  &  &  &  &  &  &  \\
     &  & -2\pi^2 &  &  &  &  &  &  \\
     &  &  & -\pi^2 &  &  &  &  &  \\
     &  &  & 1 & 1 & 1 &  &  &  \\
     &  &  &  &  & -\pi^2 &  &  &  \\
     &  &  &  &  &  & -2\pi^2 &  &  \\
     &  &  &  &  &  &  & -\pi^2 &  \\
     &  &  &  &  &  &  &  & -2\pi^2 \\
  \end{pmatrix}
  \begin{pmatrix}
    c_{0,0} \\
    c_{0,1} \\
    c_{0,2} \\
    c_{1,0} \\
    c_{1,1} \\
    c_{1,2} \\
    c_{2,0} \\
    c_{2,1} \\
    c_{2,2} \\
  \end{pmatrix}
=
  \begin{pmatrix}
    f_{0,0} \\
    f_{0,1} \\
    f_{0,2} \\
    f_{1,0} \\
    \gamma \\
    f_{1,2} \\
    f_{2,0} \\
    f_{2,1} \\
    f_{2,2} \\
  \end{pmatrix},
\end{equation}
where the additional linear equation comes from $u(0,0)=\sum_{k_1=0}^2\sum_{k_2=0}^2 c_{k_1,k_2}=\gamma$. In some problems, we might be directly given the value of $c_{1,1}$, say, $c_{1,1}=\gamma$. Then the linear system would be
\begin{equation}
  \begin{pmatrix}
    -2\pi^2 &  &  &  &  &  &  &  &  \\
     & -\pi^2 &  &  &  &  &  &  &  \\
     &  & -2\pi^2 &  &  &  &  &  &  \\
     &  &  & -\pi^2 &  &  &  &  &  \\
     &  &  &  & 1 &  &  &  &  \\
     &  &  &  &  & -\pi^2 &  &  &  \\
     &  &  &  &  &  & -2\pi^2 &  &  \\
     &  &  &  &  &  &  & -\pi^2 &  \\
     &  &  &  &  &  &  &  & -2\pi^2 \\
  \end{pmatrix}
  \begin{pmatrix}
    c_{0,0} \\
    c_{0,1} \\
    c_{0,2} \\
    c_{1,0} \\
    c_{1,1} \\
    c_{1,2} \\
    c_{2,0} \\
    c_{2,1} \\
    c_{2,2} \\
  \end{pmatrix}
=
  \begin{pmatrix}
    f_{0,0} \\
    f_{0,1} \\
    f_{0,2} \\
    f_{1,0} \\
    \gamma \\
    f_{1,2} \\
    f_{2,0} \\
    f_{2,1} \\
    f_{2,2} \\
  \end{pmatrix}.
\end{equation}

We now present the quantum Chebyshev spectral method to solve \eq{Poisson_eq} with non-periodic conditions \eq{non_period_cond}. Similarly, we choose $n=3$ in the specification of the linear system. The truncated Chebyshev series of the solution can be written as
\begin{equation}
u(x_1,x_2) = \sum_{k_1=0}^3\sum_{k_2=0}^3 c_{k_1,k_2}T_{k_1}(x_1)T_{k_2}(x_2).
\label{eq:non_period_expand}
\end{equation}
We are given an oracle for preparing the state
\begin{equation}
\sum_{l_1=0}^3\sum_{l_2=0}^3f\left(\cos\frac{\pi l_1}{3},\cos\frac{\pi l_2}{3}\right)|l_1\rangle|l_2\rangle
\label{eq:non_period_inter}
\end{equation}
that interpolates the Chebyshev-Gauss-Lobatto quadrature nodes specified in \eq{interpolation_nodes}.
We first perform the multi-dimensional inverse QCT on \eq{period_inter} to obtain \eq{Poisson_coeff}, where $f_{k_1,k_2}$ are defined by \eq{f_coeff}. Then we apply the quantum linear system algorithm of Reference~\cite{CKS15} to solve the linear system \eq{linear_system_Poisson} with the solution \eq{solution_Poisson}. The discretized linear system from \eq{Poisson_eq} is \eq{Poisson_matrix}, where the Chebyshev difference matrix $D_n$ is defined by \eq{Chebyshev_Dn} with $n=3$, namely
\begin{equation}
D_3=
  \begin{pmatrix}
    0 & 1 & 0 & 3 \\
    0 & 0 & 4 & 0 \\
    0 & 0 & 0 & 6 \\
    0 & 0 & 0 & 0 \\
  \end{pmatrix},
\end{equation}
and
\begin{equation}
D_3^2=
  \begin{pmatrix}
    0 & 0 & 4 & 0 \\
    0 & 0 & 0 & 24 \\
    0 & 0 & 0 & 0 \\
    0 & 0 & 0 & 0 \\
  \end{pmatrix}.
\end{equation}
Notice that the rank of $D_n^2$ is $n-1$, which implies the second derivative for $d=1$ can be represented as
\begin{equation}
u''(x) = c''_0T_0(x)+c''_1T_1(x) = 4c_2T_0(x)+24c_3T_1(x),
\end{equation}
where the truncation order of $u''(x)$ is $n-2$, and the coefficients $c''_0,\ldots,c''_{n-2}$ are determined by $c_2,\ldots,c_n$. Similarly for the case $d=2$, the coefficients of $\Delta u(x)$ are determined by
\begin{equation}
\begin{aligned}
c''_{00} &= 4c_{02}+4c_{20}, &
c''_{01} &= 24c_{03}+4c_{21}, &
c''_{02} &= 4c_{22}, &
c''_{03} &= 4c_{23}, \\
c''_{10} &= 4c_{12}+24c_{30}, &
c''_{11} &= 24c_{13}+24c_{31}, &
c''_{12} &= 24c_{32}, &
c''_{13} &= 24c_{33}, \\
c''_{20} &= 4c_{22}, &
c''_{21} &= 24c_{23},\\
c''_{30} &= 4c_{32}, &
c''_{33} &= 24c_{33},
\end{aligned}
\end{equation}
so the matrix $D^{2}_3\otimes I+I\otimes D^{2}_3$ gives the linear system
\begin{equation}
  \begin{pmatrix}
     &  & 4 &  &  &  &  &  & 4 &  &  &  &  &  &  &  \\
     &  &  & 24 &  &  &  &  &  & 4 &  &  &  &  &  &  \\
     &  &  &  &  &  &  &  &  &  & 4 &  &  &  &  &  \\
     &  &  &  &  &  &  &  &  &  &  & 4 &  &  &  &  \\
     &  &  &  &  &  & 4 &  &  &  &  &  & 24 &  &  &  \\
     &  &  &  &  &  &  & 24 &  &  &  &  &  & 24 &  &  \\
     &  &  &  &  &  &  &  &  &  &  &  &  &  & 24 &  \\
     &  &  &  &  &  &  &  &  &  &  &  &  &  &  & 24 \\
     &  &  &  &  &  &  &  &  &  & 4 &  &  &  &  &  \\
     &  &  &  &  &  &  &  &  &  &  & 24 &  &  &  &  \\
     &  &  &  &  &  &  &  &  &  &  &  &  &  & 4 &  \\
     &  &  &  &  &  &  &  &  &  &  &  &  &  &  & 24 \\
  \end{pmatrix}
  \begin{pmatrix}
    c_{0,0} \\
    c_{0,1} \\
    c_{0,2} \\
    c_{0,3} \\
    c_{1,0} \\
    c_{1,1} \\
    c_{1,2} \\
    c_{1,3} \\
    c_{2,0} \\
    c_{2,1} \\
    c_{2,2} \\
    c_{2,3} \\
    c_{3,0} \\
    c_{3,1} \\
    c_{3,2} \\
    c_{3,3} \\
  \end{pmatrix}
=
  \begin{pmatrix}
    f_{0,0} \\
    f_{0,1} \\
    f_{0,2} \\
    f_{0,3} \\
    f_{1,0} \\
    f_{1,1} \\
    f_{1,2} \\
    f_{1,3} \\
    f_{2,0} \\
    f_{2,1} \\
    f_{3,0} \\
    f_{3,1} \\
  \end{pmatrix}.
\label{eq:diff_system}
\end{equation}
We now use the boundary conditions to complete the linear system. The truncated Chebyshev series of the solution can be written as
\begin{equation}
\begin{aligned}
\gamma_N(x_1) &= \sum_{k_1=0}^2 {g_N}_{k_1}T_{k_1}(x_1),\\
\gamma_S(x_1) &= \sum_{k_1=0}^2 {g_S}_{k_1}T_{k_1}(x_1),\\
\gamma_E(x_2) &= \sum_{k_2=0}^2 {g_E}_{k_1}T_{k_1}(x_2),\\
\gamma_W(x_2) &= \sum_{k_2=0}^2 {g_W}_{k_1}T_{k_1}(x_2).
\label{eq:boundary_expand}
\end{aligned}
\end{equation}
We are given an oracle for preparing the state by interpolating the Chebyshev-Gauss-Lobatto quadrature nodes specified in \eq{interpolation_nodes}
\begin{equation}
\begin{aligned}
&\sum_{l_1=0}^3\gamma_N\left(\cos\frac{\pi l_1}{3}\right)|l_1\rangle, &
&\sum_{l_2=0}^3\gamma_E\left(\cos\frac{\pi l_2}{3}\right)|l_2\rangle,\\
&\sum_{l_1=0}^3\gamma_S\left(\cos\frac{\pi l_1}{3}\right)|l_1\rangle, &
&\sum_{l_2=0}^3\gamma_W\left(\cos\frac{\pi l_2}{3}\right)|l_2\rangle.
\label{eq:boundary_inter}
\end{aligned}
\end{equation}
We perform the multi-dimensional inverse QCT on \eq{boundary_inter} to obtain
\begin{equation}
\begin{aligned}
&\sum_{k_1=0}^3{g_N}_{k_1}|k_1\rangle, &
&\sum_{k_2=0}^3{g_E}_{k_2}|k_2\rangle, \\
&\sum_{k_1=0}^3{g_S}_{k_1}|k_1\rangle, &
&\sum_{k_2=0}^3{g_W}_{k_2}|k_2\rangle,\\
\label{eq:boundary_coeff}
\end{aligned}
\end{equation}
where $a_{k_1,k_2}$ are defined by \eq{g_coeff}. The linear system from the boundary conditions is
\begin{equation}
\small
  \begin{pmatrix}
    1 & 1 & 1 & 1 &  &  &  &  &  &  &  &  &  &  &  &  \\
     &  &  &  & 1 & 1 & 1 & 1 &  &  &  &  &  &  &  &  \\
     &  &  &  &  &  &  &  & 1 & 1 & 1 & 1 &  &  &  &  \\
     &  &  &  &  &  &  &  &  &  &  &  & 1 & 1 & 1 & 1 \\
    1 & -1 & 1 & -1 &  &  &  &  &  &  &  &  &  &  &  &  \\
     &  &  &  & 1 & -1 & 1 & -1 &  &  &  &  &  &  &  &  \\
     &  &  &  &  &  &  &  & 1 & -1 & 1 & -1 &  &  &  &  \\
     &  &  &  &  &  &  &  &  &  &  &  & 1 & -1 & 1 & -1 \\
    1 &  &  &  & 1 &  &  &  & 1 &  &  &  & 1 &  &  &  \\
     & 1 &  &  &  & 1 &  &  &  & 1 &  &  &  & 1 &  &  \\
     &  & 1 &  &  &  & 1 &  &  &  & 1 &  &  &  & 1 &  \\
     &  &  & 1 &  &  &  & 1 &  &  &  & 1 &  &  &  & 1 \\
    1 &  &  &  & -1 &  &  &  & 1 &  &  &  & -1 &  &  &  \\
     & 1 &  &  &  & -1 &  &  &  & 1 &  &  &  & -1 &  &  \\
     &  & 1 &  &  &  & -1 &  &  &  & 1 &  &  &  & -1 &  \\
     &  &  & 1 &  &  &  & -1 &  &  &  & 1 &  &  &  & -1 \\
  \end{pmatrix}
  \begin{pmatrix}
    c_{0,0} \\
    c_{0,1} \\
    c_{0,2} \\
    c_{0,3} \\
    c_{1,0} \\
    c_{1,1} \\
    c_{1,2} \\
    c_{1,3} \\
    c_{2,0} \\
    c_{2,1} \\
    c_{2,2} \\
    c_{2,3} \\
    c_{3,0} \\
    c_{3,1} \\
    c_{3,2} \\
    c_{3,3} \\
  \end{pmatrix}
=
  \begin{pmatrix}
    {g_N}_{0} \\
    {g_N}_{1} \\
    {g_N}_{2} \\
    {g_N}_{3} \\
    {g_S}_{0} \\
    {g_S}_{1} \\
    {g_S}_{2} \\
    {g_S}_{3} \\
    {g_E}_{0} \\
    {g_E}_{1} \\
    {g_E}_{2} \\
    {g_E}_{3} \\
    {g_W}_{0} \\
    {g_W}_{1} \\
    {g_W}_{2} \\
    {g_W}_{3} \\
  \end{pmatrix}.
\label{eq:boundary_system}
\end{equation}
Adding the two linear systems \eq{diff_system} and \eq{boundary_system} together, we obtain a full-rank linear system
\begin{equation}
\overline{D}^{(2)}_3\otimes I+I\otimes \overline{D}^{(2)}_3,
\label{eq:complete_matrix}
\end{equation}
where
\begin{equation}
\overline{D}_3^{(2)}=
  \begin{pmatrix}
    0 & 0 & 4 & 0 \\
    0 & 0 & 0 & 24 \\
    1 & -1 & 1 & -1 \\
    1 & 1 & 1 & 1 \\
  \end{pmatrix}.
\end{equation}
In summary, the linear system including the differential equations and the boundary conditions is
\begin{equation}
  \footnotesize
  \setlength\arraycolsep{4pt}
  \begin{pmatrix}
     &  & 4 &  &  &  &  &  & 4 &  &  &  &  &  &  &  \\
     &  &  & 24 &  &  &  &  &  & 4 &  &  &  &  &  &  \\
    1 & -1 & 1 & -1 &  &  &  &  &  &  & 4 &  &  &  &  &  \\
    1 & 1 & 1 & 1 &  &  &  &  &  &  &  & 4 &  &  &  &  \\
     &  &  &  &  &  & 4 &  &  &  &  &  & 24 &  &  &  \\
     &  &  &  &  &  &  & 24 &  &  &  &  &  & 24 &  &  \\
     &  &  &  & 1 & -1 & 1 & -1 &  &  &  &  &  &  & 24 &  \\
     &  &  &  & 1 & 1 & 1 & 1 &  &  &  &  &  &  &  & 24 \\
    1 &  &  &  & -1 &  &  &  & 1 &  & 4 &  & -1 &  &  &  \\
     & 1 &  &  &  & -1 &  &  &  & 1 &  & 24 &  & -1 &  &  \\
     &  & 1 &  &  &  & -1 &  & 1 & -1 & 2 & -1 &  &  & -1 &  \\
     &  &  & 1 &  &  &  & -1 & 1 & 1 & 1 & 2 &  &  &  & -1 \\
    1 &  &  &  & 1 &  &  &  & 1 &  &  &  & 1 &  & 4 &  \\
     & 1 &  &  &  & 1 &  &  &  & 1 &  &  &  & 1 &  & 24 \\
     &  & 1 &  &  &  & 1 &  &  &  & 1 &  & 1 & -1 & 2 & -1 \\
     &  &  & 1 &  &  &  & 1 &  &  &  & 1 & 1 & 1 & 1 & 2 \\
  \end{pmatrix}
    \begin{pmatrix}
    c_{0,0} \\
    c_{0,1} \\
    c_{0,2} \\
    c_{0,3} \\
    c_{1,0} \\
    c_{1,1} \\
    c_{1,2} \\
    c_{1,3} \\
    c_{2,0} \\
    c_{2,1} \\
    c_{2,2} \\
    c_{2,3} \\
    c_{3,0} \\
    c_{3,1} \\
    c_{3,2} \\
    c_{3,3} \\
  \end{pmatrix}
=
  \begin{pmatrix}
    f_{0,0} \\
    f_{0,1} \\
    f_{0,2}+{g_S}_{0} \\
    f_{0,3}+{g_N}_{0} \\
    f_{1,0} \\
    f_{1,1} \\
    f_{1,2}+{g_S}_{1} \\
    f_{1,3}+{g_N}_{1} \\
    f_{2,0}+{g_W}_{0} \\
    f_{2,1}+{g_W}_{1} \\
    {g_W}_{2}+{g_S}_{2} \\
    {g_W}_{3}+{g_N}_{2} \\
    f_{3,0}+{g_E}_{0} \\
    f_{3,1}+{g_E}_{1} \\
    {g_E}_{2}+{g_S}_{3} \\
    {g_E}_{3}+{g_N}_{3} \\
  \end{pmatrix}.
\label{eq:complete_system}
\end{equation}


\section{Singular values of second-order differential matrices}
\label{app:svd}

Here we present a detailed proof of the singular value estimation in \lem{svd_Fourier} and \lem{svd_Chebyshev}.

\svdFourier*

\begin{proof}
By direct calculation of the $l_{\infty}$ norm (i.e., the maximum absolute column sum) of \eq{matrix_Fourier}, for $n\ge4$, we have
\begin{equation}
\|\overline{D}^{(2)}_n\|_{\infty} \le \left(\frac{(n+1)\pi}{2}\right)^2 \le (2n)^2.
\end{equation}

Then the inverse of the matrix \eq{matrix_Fourier} is
\begin{equation}
\begin{aligned}
&[(\overline{D}^{(2)}_n)^{-1}]_{k,k}=-\frac{1}{((k-\left\lfloor n/2 \right\rfloor)\pi)^{2}},\quad &k \in\rangez{n+1}\backslash\{\left\lfloor n/2 \right\rfloor\},\\
&[(\overline{D}^{(2)}_n)^{-1}]_{\left\lfloor n/2 \right\rfloor,k}=\frac{1}{((k-\left\lfloor n/2 \right\rfloor)\pi)^{2}},\quad &k \in\rangez{n+1}
\end{aligned}
\end{equation}
as can easily be verified by a direct calculation.

By direct calculation of the Frobenius norm of \eq{matrix_Fourier}, we have
\begin{equation}
\|(\overline{D}^{2}_n)^{-1}\|_F^2 \le 1+2\sum_{k=1}^{\infty}\frac{1}{k^4\pi^4} = 1+\frac{2}{\pi^4}\frac{\pi^4}{90} \le 2.
\end{equation}

Thus we have the result in \eq{svd_Fourier}:
\begin{equation}
\begin{aligned}
&\sigma_{\max}(\overline{D}^{(2)}_n) \le \sqrt{n+1}\|\overline{D}^{2}_n\|_{\infty} \le (2n)^{2.5},\\
&\sigma_{\min}(\overline{D}^{(2)}_n) \ge \frac{1}{\|(\overline{D}^{2}_n)^{-1}\|_F} \ge \frac{1}{\sqrt{2}}
\end{aligned}
\end{equation}
as claimed.
\end{proof}

\svdChebyshev*

\begin{proof}
By direct calculation of the Frobenius norm of \eq{matrix_Chebyshev}, we have
\begin{equation}
\|\overline{D}^{(2)}_n\|_F^2 \le n^2\max_{k,r}\left(\frac{r(r^2-k^2)}{\sigma_k}\right)^2 \le n^2\cdot n^6 = n^8.
\end{equation}

Next we upper bound $\|(\overline{D}^{(2)}_n)^{-1}\|$. By definition,
\begin{equation}
\|(\overline{D}^{(2)}_n)^{-1}\|=\sup_{\|b\|\le1}\|(\overline{D}^{2}_n)^{-1}b\|.
\end{equation}
Given any vector $b$ satisfying $\|b\|\le1$, we estimate $\|x\|$ defined by the full-rank linear system
\begin{equation}
\overline{D}^{(2)}_n x = b.
\label{eq:svd_system}
\end{equation}
Notice that $\overline{D}^{(2)}_n$ is the sum of the upper triangular matrix ${D}^{2}_n $ and \eq{Chebyshev_Gn}, the coordinates $x_2,\ldots,x_n$ are only defined by coordinates $b_0,\ldots,b_{n-2}$. So we only focus on the partial system
\begin{equation}
{D}^{(2)}_n [0,0,x_2,\ldots,x_n]^T = [b_0,\ldots,b_{n-2},0,0]^T.
\end{equation}
Given the same $b$, we also define the vector $y$ by
\begin{equation}
{D}_n [0,y_1,\ldots,y_{n-1},0]^T = [b_0,\ldots,b_{n-2},0,0]^T,
\end{equation}
where each coordinate of $y$ can be expressed by
\begin{equation}
b_k = \sum_{l=1}^{n-1}[{D}_n]_{kl}y_l = \sum_{l=1}^{n-1}\frac{2l}{\sigma_k}y_l, \quad \text{$k+l$ odd},~ l>k,~k\in\rangez{n-1}.
\label{eq:svd_by}
\end{equation}
Using this equation with $k=l-1$ and $k=l+1$, we can express $y_l$ in terms of $b_{l-1}$ and $b_{l+1}$:
\begin{equation}
\frac{2l}{\sigma_{l-1}} y_l = b_{l-1} - \frac{1}{\sigma_{l-1}} b_{l+1}, \quad l \in\range{n-1},
\label{eq:svd_yb}
\end{equation}
where we let $b_{n-1}=b_n=0$. Thus we have
\begin{equation}
\begin{aligned}
\sum_{l=1}^{n-1}y_l^2
&= \sum_{l=1}^{n-1} \left(\frac{\sigma_{l-1}}{2l}\left(b_{l-1} - \frac{1}{\sigma_{l-1}} b_{l+1}\right)\right)^2 \\
&\le \sum_{l=1}^{n-1}\frac{\sigma_{l-1}^2}{4l^2}\left(1+\frac{1}{\sigma_{l-1}^2}\right)(b_{l-1}^2+b_{l+1}^2) \\
&\le \frac{5}{4}(b_0^2+b_2^2)+\frac{1}{16}\sum_{l=2}^{n-2}(b_{l-1}^2+b_{l+1}^2) \\
&\le 2\sum_{l=0}^{n-2}b_l^2.
\label{eq:svd_y2b2}
\end{aligned}
\end{equation}
We notice that $y$ also satisfies
\begin{equation}
[0,y_1,\ldots,y_{n-1},0]^T = {D}_n [0,0,x_2,\ldots,x_n]^T,
\end{equation}
where each coordinate of $y$ can be expressed by
\begin{equation}
y_l = \sum_{r=1}^n[{D}_n]_{lr}x_r = \sum_{r=1}^n\frac{2r}{\sigma_l}x_r, \quad \text{$l+r$ odd},~ r>l,~l\in\range{n-1}.
\label{eq:svd_yx}
\end{equation}
Substituting the $(r-1)$st and the $(r+1)$st equations of \eq{svd_yx}, we can express $x$ in terms of $y$:
\begin{equation}
\frac{2r}{\sigma_{r-1}} x_r = y_{r-1} - \frac{1}{\sigma_{r-1}} y_{r+1}, \quad r \in\range{n} \backslash \{1\},
\label{eq:svd_xy}
\end{equation}
where we let $y_n=y_{n+1}=0$. Similarly, according to \eq{svd_xy}, we also have
\begin{equation}
\sum_{l=2}^nx_l^2\le2\sum_{l=1}^{n-1}y_l^2.
\label{eq:svd_x2y2}
\end{equation}
Then we calculate $x_0^2+x_1^2$ based on the last two equations of \eq{svd_system}, \eq{svd_y2b2}, and \eq{svd_xy}, giving
\begin{equation}
\begin{aligned}
x_0^2+x_1^2
&= \frac{1}{2}[(x_0+x_1)^2+(x_0-x_1)^2]\\
&= \frac{1}{2}\left[\left(b_n-\sum_{l=2}^nx_l\right)^2+\left(b_{n-1}-\sum_{l=2}^n(-1)^lx_l\right)^2\right] \\
&= \frac{1}{2}\left[\left(b_n-\sum_{l=2}^n\frac{\sigma_{l-1}}{2l}\left(y_{l-1} - \frac{1}{\sigma_{l-1}} y_{l+1}\right)\right)^2 \right. \\
&\qquad+\left.\left(b_{n-1}-\sum_{l=2}^n(-1)^l\frac{\sigma_{l-1}}{2l}\left(y_{l-1} - \frac{1}{\sigma_{l-1}} y_{l+1}\right)\right)^2\right] \\
&\le \frac{1}{2}\left(1+\sum_{l=2}^n\frac{\sigma_{l-1}^2}{4l^2}\right)\left[b_n^2+\sum_{l=2}^n\left(y_{l-1} - \frac{1}{\sigma_{l-1}} y_{l+1}\right)^2+b_{n-1}^2\right.\\
&\qquad+\left.\sum_{l=2}^n\left(y_{l-1} - \frac{1}{\sigma_{l-1}} y_{l+1}\right)^2\right] \\
&\le \frac{1}{2}\left(1+\frac{1}{4}\sum_{l=2}^{\infty}\frac{1}{l^2}\right)\left[b_n^2+b_{n-1}^2+\sum_{l=2}^{n}\left(1+\frac{1}{\sigma_{l-1}^2}\right)(y_{l-1}^2+y_{l+1}^2)\right] \\
&\le \frac{1}{2}\left(1+\frac{\pi^2}{24}\right)\left[b_n^2+b_{n-1}^2+4\sum_{l=1}^{n-1}y_l^2\right] \\
&\le b_n^2+b_{n-1}^2+8\sum_{l=0}^{n-2}b_l^2.
\label{eq:svd_x0x1}
\end{aligned}
\end{equation}
Thus, based on \eq{svd_y2b2}, \eq{svd_x2y2}, and \eq{svd_x0x1}, the inequality
\begin{equation}
\begin{aligned}
\sum_{l=0}^nx_l^2 &= x_0^2+x_1^2+\sum_{l=2}^nx_l^2 \\
&\le b_n^2+b_{n-1}^2+8\sum_{l=0}^{n-2}b_l^2+4\sum_{l=0}^{n-2}b_l^2 \\
&\le b_n^2+b_{n-1}^2+12\sum_{l=0}^{n-2}b_l^2 \le 12
\end{aligned}
\end{equation}
holds for any vectors $b$ satisfying $\|b\|\le1$. Thus
\begin{equation}
\|(\overline{D}^{(2)}_n)^{-1}\| = \sup_{\|b\|\le1} \|x\| \le 12 < 16.
\end{equation}

Altogether, we have
\begin{equation}
\begin{aligned}
&\sigma_{\max}(\overline{D}^{(2)}_n) \le \|\overline{D}^{2}_n\|_F \le n^4,\\
&\sigma_{\min}(\overline{D}^{(2)}_n) \ge \frac{1}{\|(\overline{D}^{2}_n)^{-1}\|} \ge \frac{1}{16}
\end{aligned}
\end{equation}
as claimed in \eq{svd_Chebyshev}.
\end{proof}

\end{document}